\documentclass[a4paper,10pt]{article}

\usepackage{geometry} 

\usepackage{comment}

\usepackage[T1]{fontenc}
\usepackage[utf8]{inputenc}

\usepackage{amssymb, amsthm, amsmath}

\usepackage{authblk}
\usepackage[hidelinks]{hyperref} 
\usepackage[dvipsnames]{xcolor}
\usepackage[biblabel]{cite}

\usepackage{mathdots}

\usepackage{tikz-cd}

\usepackage{bm,xcolor,tikz,hyperref, float}

\newtheorem{theorem}{Theorem}[section]
\newtheorem{lemma}[theorem]{Lemma}
\newtheorem{proposition}[theorem]{Proposition}
\newtheorem{corollary}[theorem]{Corollary}

\theoremstyle{definition}
\newtheorem{definition}[theorem]{Definition} 
\newtheorem{remark}[theorem]{Remark}
\newtheorem{example}[theorem]{Example}

\newcommand{\wt}{{\rm wt}}
\newcommand{\dd}{{\rm d}}
\newcommand{\C}{\mathcal{C}}
\newcommand{\M}{\mathcal{M}}
\newcommand{\V}{\mathcal{V}}
\newcommand{\K}{\mathcal{K}}
\newcommand{\A}{\mathcal{A}}
\newcommand{\F}{\mathbb{F}}
\newcommand{\N}{\mathbb{N}}
\newcommand{\ones}{\mathbf{1}}

\newcommand{\GL}{{\rm GL}}
\newcommand{\PGL}{{\rm PGL}}
\newcommand{\Aut}{{\rm Aut}}
\newcommand{\PAut}{{\rm PAut}}
\newcommand{\FAut}{{\rm FAut}}

\newcommand{\diag}{{\rm diag}}
\newcommand{\ee}{\mathbf{e}}
\newcommand{\cc}{\mathbf{c}}
\newcommand{\mm}{\mathbf{m}}

\title{The permutation and automorphism groups of\\univariate multiplicity codes}
\author[1]{Eduardo Camps Moreno}
\author[2]{Hiram H. L{\'o}pez}
\author[3]{Umberto Mart{\'i}nez-Pe\~{n}as}
\author[2]{Rodrigo San-Jos{\'e}}
\affil[1]{Institut de Mathématiques de Bordeaux, France\footnote{Email: eduardo.camps-moreno@math.u-bordeaux.fr}}
\affil[2]{Department of Mathematics\\Virginia Tech, VA, USA\footnote{Emails: \{hhlopez,rsanjose\}@vt.edu}}
\affil[3]{IMUVa-Mathematics Research Institute\\University of Valladolid, Spain\footnote{Email: umberto.martinez@uva.es}}

\date{}

\begin{document}
\maketitle
\begin{abstract}
We give explicit characterizations of the permutations and linear isometries that leave (univariate) multiplicity codes invariant. We consider any evaluation subset of the finite field, any multiplicity and dimension, and the automorphisms are with respect to the Niederreiter-Rosenblom-Tsfasman (NRT) metric. For multiplicity one, we recover the permutation and automorphism groups of classical Reed--Solomon codes in the classical Hamming metric. In the main dimension range, automorphisms are uniquely determined by a projectivity that permutes the evaluation subset, up to multiplication by a nonzero scalar. For multiplicity one, permutations are given by affine maps that permute the evaluation subset. In contrast, for multiplicity at least two, permutations are only given by translations that permute the evaluation subset. Finally, we give an algebraic characterization of automorphisms of multiplicity codes for the folded Hamming metric.


\end{abstract}

\section{Introduction} \label{sec intro}

The automorphism groups of classical Reed--Solomon codes \cite{RS} in the Hamming metric were obtained in \cite[Th. 4]{Dur1987}. A natural generalization of Reed--Solomon codes, obtained by evaluating also derivatives, was first defined (for a single variable) by Rosenbloom and Tsfasman in \cite{RT1997}. There they also introduce the $ m $-metric, nowadays known as the Niederreiter-Rosenblom-Tsfasman (or NRT) metric. Such codes were later generalized to several variables in \cite{KSY2014}, where they are called multiplicity codes. However, there they are considered with respect to the folded Hamming metric. For this metric, (univariate) multiplicity codes achieve the optimal rate versus list-decodability tradeoff \cite{guruswami2013linear}. Later, multiplicity codes have independently been studied under the name hyperderivative Reed--Solomon (HRS) codes \cite{CanHorowitz2026,GuZhang2026}. Multiplicity codes (for one variable) attain the Singleton bound for both the NRT metric and the folded Hamming metric. Both metrics particularize to the classical Hamming metric for multiplicity one (the Reed--Solomon case).

To the best of our knowledge, the permutation and automorphism groups of multiplicity codes have not yet been investigated. Whereas there is a clear notion of permutation group, automorphisms could be studied for the NRT metric or the folded Hamming metric. Linear NRT isometries are given by permutations and invertible upper triangular matrices, whereas linear folded-Hamming isometries are given by permutations and arbitrary invertible matrices. 

In this manuscript, we explicitly describe the NRT automorphisms of multiplicity codes since they have exactly the same group structure as the Hamming-metric automorphism group of classical Reed--Solomon codes (in the central dimension range), as we will show. However, we also briefly discuss an algebraic characterization of the automorphism group of multiplicity codes for the folded Hamming metric, which is generally much larger. 

We treat univariate multiplicity codes of arbitrary multiplicity and dimension and evaluated over arbitrary subsets of a finite field. For most dimensions, automorphisms for the NRT metric are given by projectivities that stabilize the evaluation set and a global nonzero scalar. Permutations are given by affine maps that stabilize the evaluation set for multiplicity one (Reed--Solomon codes), and only affine translations for multiplicity at least two. Finally, folded Hamming automorphisms are given by those whose associated polynomial map stabilizes (after permutation) the subspaces of polynomials with prescribed multiplicities.

The manuscript is organized as follows. Section \ref{sec preliminaries} contains preliminaries on multiplicity codes, the NRT metric and the finite projective line. Section \ref{sec projectivities autom} shows that projectivities that stabilize the evaluation set yield NRT automorphisms. Section \ref{sec permutation is projectivity} shows that for an NRT automorphism, the permutation is necessarily given by a projectivity. Section \ref{sec d degree graph} shows that, when the permutation is the identity, an NRT automorphism consists simply in multiplying by a global nonzero scalar. Section \ref{sec classification} collects these results to fully characterize the NRT automorphism group of multiplicity codes. Section \ref{sec permutation} uses such an NRT characterization to describe the permutation group. Finally, Section \ref{sec folded Hamming} gives an algebraic characterization of the automorphisms for the folded Hamming metric. The appendices give the low and high dimensional ranges for the NRT automorphism groups, and the extreme cases of dimension one, codimension one, and one or two evaluation points.

\section{Preliminaries} \label{sec preliminaries}

\textbf{Basic notation.} Throughout this manuscript, $ \F_q $ denotes the finite field of size $ q $, which is a power of a prime. We set $ \mathbb{N} = \{ 0,1,2, \ldots \} $. For $ k \in \mathbb{N} $, we denote by $ \F_q[x] $ the ring of polynomials over $ \F_q $, and by $ \F_q[x]_{< k} $ and $ \F_q[x]_{\leq k} $ the sets of polynomials in $ \F_q[x] $ of degree less than $ k $ and at most $ k $, respectively (where $ \F_q[x]_{<0} = \{0\} $). We indistinctively write $ f $ and $ f(x) $ for polynomials in $ \F_q[x] $. We denote by $ \langle \A \rangle_{\F_q} $ the vector space over $ \F_q $ generated by a set $ \A $. For integers $ m,n $, we denote $ [m,n] = \{ m,m+1, \ldots, n\} $ (if $ m \leq n $) and $ [n] = [1,n] $, we let $ \F_q^{m \times n} $ be the set of $ m \times n $ matrices over $ \F_q $, and we define $ \F_q^n = \F_q^{1 \times n} $.

Throughout this paper, we fix an arbitrary nonempty evaluation set $\A \subseteq \F_q$, where we denote $n = |\A| > 0$.

\subsection{Multiplicity codes} \label{subsec multiplicity codes}

As usual in the multiplicity code literature, we will consider Hasse derivatives in this work.

\begin{definition} [\textbf{Hasse derivatives} \cite{hasse}] \label{def hasse derivative}
Take $ f \in \F_q[x] $. Given a variable $ y $, algebraically independent of $ x $, the polynomial $ f(x+y) \in \F_q[x,y] $ can be written uniquely as
$$ f(x+y) = \sum_{i \in \N} f^{(i)}(x) y^i, $$
for some polynomials $ f^{(i)}(x) \in \F_q[x] $, for $ i \in \N $. We define the $ i $th Hasse derivative of $ f(x) \in \F_q[x] $ as the polynomial $ f^{(i)}(x) \in \F_q[x] $, for $ i \in \N $.
\end{definition}

Notice that $ f^{(0)} = f $, and $ f^{(1)} = f^\prime $ is the usual derivative of $ f \in \F_q[x] $. The following relationship between Hasse derivatives and multiplicities of roots of polynomials is trivial and holds even though $ \F_q $ has positive characteristic (that is why Hasse derivatives are used).

\begin{proposition} \label{prop multiplicity hasse derivatives}
Given $ f \in \F_q[x] $, $ a \in \F_q $ and a positive integer $ m $, then $ (x-a)^m $ divides $ f $ if and only if $ f^{(0)}(a) = f^{(1)}(a) = \ldots = f^{(m-1)}(a) = 0 $.
\end{proposition}

We may define (univariate) multiplicity codes as follows. For the rest of the manuscript, we fix integers $r\geq 1$ and $ k \in [rn]$.

\begin{definition}[\textbf{Multiplicity codes} \cite{RT1997}]
For $ f \in \F_q[x] $, we define its $ r $th jet on $ a \in \F_q $ as
$$ J^{(r)}_a(f) = (f(a), f^{(1)}(a), \ldots, f^{(r-1)}(a)) \in \F_q^r. $$
Then we define the $ k $-dimensional multiplicity code of multiplicity $ r $ over the evaluation set $ \A $ as
$$ \M_\A(r,k) = \left\lbrace \left( J^{(r)}_a(f) \right)_{a \in \A} : f \in \V_k \right\rbrace \subseteq (\F_q^r)^n, $$
where we will repeatedly use the abbreviated notation
$ \V_k = \F_q[x]_{< k}$.
\end{definition}

Technically, we should write $ (\F_q^r)^\A $ instead of $ (\F_q^r)^n $, but we implicitly assume some ordering of $ \A $.

Clearly, $ \M_\A(1,k) $ (i.e., multiplicity codes for $ r=1 $) correspond to classical Reed--Solomon codes \cite{RS}. The following result follows since the polynomials $ x-a $, for different $ a \in \F_q $, are pairwise coprime and $ k \leq rn $. 

\begin{proposition} \label{prop ev is injective}
The evaluation map
$$ {\rm ev}_\A : \V_k \longrightarrow \M_\A(r,k), $$
which sends $ f \in \V_k $ to $ (J^{(r)}_a(f))_{a \in \A} $, is a vector space isomorphism. In particular,
$$ \dim(\M_\A(r,k)) = k. $$
\end{proposition}

\subsection{The NRT metric} \label{subsec NRT metric}

In this article, we will consider multiplicity codes with respect to the NRT metric, which was introduced in \cite[Sec. 1]{RT1997}.

\begin{definition} [\textbf{NRT metric} \cite{RT1997}]
For $ z = (z_0,z_1, \ldots, z_{r-1}) \in \F_q^r $, we define its NRT weight as
$$ \wt(z) = \left\lbrace
\begin{array}{ll}
0, & \textrm{ if } z = 0, \\
r-j, & \textrm{ if } j = \min \{ i \in [0,r-1] : z_i \neq 0 \}. 
\end{array}  \right. $$
For $ c = (c_a)_{a \in \A} \in (\F_q^r)^n $, we define its NRT weight as
$$ \wt(c) = \sum_{a \in \A} \wt(c_a). $$
Finally we define the NRT metric by $ \dd(c,c^\prime) = \wt(c-c^\prime) $, for all $ c,c^\prime \in \F_q^r $ and $ c,c^\prime \in (\F_q^r)^n $. Context will clarify in which of the two cases we are in.

As usual, given a code $ \C \subseteq (\F_q^r)^n $, we define its minimum NRT distance as
$$ \dd(\C) = \min \{ \dd(c,c^\prime) : c,c^\prime \in \C, c \neq c^\prime \}, $$
which coincides with the minimum NRT weight of a nonzero codeword if $ \C $ is linear.
\end{definition}

Observe that, for $ r=1 $, the NRT metric coincides with the classical Hamming metric in $ \F_q^n $. The NRT metric also satisfies a Singleton bound for every $ r \geq 1 $. See \cite[Th.~1]{RT1997}.

\begin{proposition} [\textbf{Singleton bound} \cite{RT1997}]
Given a code $ \C \subseteq (\F_q^r)^n $ with $ d = \dd(\C) $ and $ k = \log_q|\C| $ (i.e., the dimension over $ \F_q $ of $ \C $ if it is linear), then
$$ d \leq rn - k + 1. $$
\end{proposition}

It was shown in \cite[Th.~6]{RT1997} that multiplicity codes attain such a Singleton bound.

\begin{proposition} [\hspace{1pt}\cite{RT1997}]
For the multiplicity code $ \M_\A(r,k) $, we have
$$ \dd(\M_\A(r,k)) = rn - k + 1. $$
\end{proposition}

As usual, we say that a map $ T : (\F_q^r)^n \longrightarrow (\F_q^r)^n $ is a linear isometry for the NRT metric if it is linear and $ \wt(T(c)) = \wt(c) $, for all $ c \in (\F_q^r)^n $. We have the following characterization of the linear isometries for the NRT metric. We will use the elements $ a \in \A $ as the subindices for tuples in $ (\F_q^r)^n $. Later this notation will prove quite useful. The following result, obtained first in \cite{Lee2003}, was generalized to poset metrics in \cite[Th. 1]{PanekEtAl2008}.

\begin{proposition} [\hspace{1pt}\cite{Lee2003}] \label{prop isometries NRT charact}
A linear map $ T : (\F_q^r)^n \longrightarrow (\F_q^r)^n $ is an isometry for the NRT metric if and only if there exists a permutation $ \pi : \A \longrightarrow \A $ and upper triangular matrices $ A_a \in \GL_r(\F_q) $, for $ a \in \A $, such that 
$$ T((z_a)_{a \in \A}) = (z_{\pi(a)} A_a)_{a \in \A}, $$
for all $ (z_a)_{a \in \A} \in (\F_q^r)^n $. 
\end{proposition}

The two main objects of our interest in this manuscript are the following two groups.

\begin{definition} [\textbf{Permutation and automorphism groups}] \label{def Aut and PAut}
We define the automorphism group of $ \M_\A(r,k) $ as the group $ \Aut(\M_\A(r,k)) $ of linear isometries $ T : (\F_q^r)^n \longrightarrow (\F_q^r)^n $ for the NRT metric such that 
$$ T(\M_\A(r,k)) = \M_\A(r,k). $$
Given a permutation $ \pi : \A \longrightarrow \A $, we define $ P_\pi : (\F_q^r)^n \longrightarrow (\F_q^r)^n $ by $ P_\pi((z_a)_{a \in \A}) = (z_{\pi^{-1}(a)})_{a \in \A} $, for all $ (z_a)_{a \in \A} \in (\F_q^r)^n $. Finally, we define the permutation group of $ \M_\A(r,k) $ as the group $ \PAut(\M_\A(r,k)) $ of maps $ P_\pi : (\F_q^r)^n \longrightarrow (\F_q^r)^n $, for permutations $ \pi : \A \longrightarrow \A $, such that
$$ P_\pi(\M_\A(r,k)) = \M_\A(r,k). $$
\end{definition}

\begin{remark}
We associate $ \pi^{-1} $ to $ P_\pi $ so that the map that sends a permutation $ \pi $ of $ \A $ to the linear isometry $ P_\pi $ is a group morphism instead of an anti-morphism (i.e., with our definition, $ P_\pi \circ P_\rho = P_{\pi \circ \rho} $, for all permutations $ \pi, \rho : \A \longrightarrow \A $).
\end{remark}

\begin{remark}
In the literature, multiplicity codes are considered sometimes with the NRT metric \cite{CanHorowitz2026, GuZhang2026, LiEtAl2026, RT1997} and sometimes with the folded Hamming metric \cite{KSY2014, martinez2025linear} (i.e., the classical Hamming metric in $ (\F_q^r)^n $ with length $ n $ and alphabet $ \F_q^r $). With the latter, they are MDS (if $ r | k $) and dually quasi MDS (otherwise), see \cite[Cor. 31]{martinez2025linear}. The permutation and automorphism groups $ \Aut(\M_A(r,k)) $ and $ \PAut(\M_A(r,k)) $ are both subgroups of the automorphism group of $ \M_A(r,k) $ for the folded Hamming metric by \cite[Th. 2]{sr-hamming} (see also \cite[Prop. 5]{martinez2025linear}). However, the results on automorphism groups for Reed--Solomon codes for the Hamming metric \cite{Dur1987} generalize very naturally to multiplicity codes in the NRT metric, as we show throughout this paper. However, we also discuss their folded Hamming metric automorphisms in Section \ref{sec folded Hamming}.
\end{remark}

\subsection{The projective line} \label{subsec projective line}

We will need some basic facts regarding the projective line over $ \F_q $, which we revisit here. For more details, we refer the reader to \cite{BeutelspacherRosenbaum1998, Hirschfeld1998}.

\begin{definition}[\textbf{Projective line}]
Define in $ \F_q^2 \setminus \{ 0 \} $ the equivalence relation $ (u_1,u_2) \sim (v_1,v_2) $ if there exists $ \lambda \in \F_q^* $ such that $ (v_1,v_2) = (\lambda u_1, \lambda u_2) $, for all $ (u_1,u_2),(v_1,v_2) \in \F_q^2 \setminus \{ 0 \} $. Denote by $ [v_1 : v_2] $ the class of $ (v_1,v_2) \in \F_q^2 \setminus \{ 0 \} $ with respect to that relation. The projective line $ \mathbb{P}^1(\F_q) $ is simply defined as the quotient $ (\F_q^2 \setminus \{ 0 \}) / \sim $, i.e., the set of equivalence classes
$$ \mathbb{P}^1(\F_q) = \{ [v_1: v_2] : (v_1,v_2) \in \F_q^2 \setminus \{ 0 \} \}. $$
\end{definition}

Observe that if $ [v_1:v_2] \in \mathbb{P}^1(\F_q) $ is such that $ v_2 \neq 0 $, then the same happens for any other representative of its class, and dividing by $ v_2 $, we obtain 
$$ [v_1:v_2] = [v_1v_2^{-1} : 1]. $$
Similarly, if $ v \neq 0 $, then $ [v:0] = [1:0] $. Thus if we denote $ \infty = [1:0] $ and identify $ a \in \F_q $ with $ [a:1] $, we may write $ \mathbb{P}^1(\F_q) $ as a disjoint union
$$ \mathbb{P}^1(\F_q) = \F_q \cup \{ \infty \}. $$

We will make use of projectivities, especially in fractional form.

\begin{definition}[\textbf{Projectivities}] \label{def projectivity}
Given
$$ M = \left( \begin{array}{cc}
\alpha & \gamma \\
\beta & \delta
\end{array} \right) \in \GL_2(\F_q), $$
we define the projectivity $ \psi_M : \mathbb{P}^1(\F_q) \longrightarrow \mathbb{P}^1(\F_q) $ by
$$ \psi_M([v_1 : v_2]) = [ (v_1,v_2) \cdot M ], $$
which does not depend on the representative $ (v_1,v_2) \in \F_q^2 \setminus \{ 0 \} $. In $ \F_q \cup \{ \infty \} $, it acts in fractional form as
$$ \psi_M(x) = \frac{\alpha x + \beta}{\gamma x + \delta}, $$
for all $ x \in \F_q $ such that $ \gamma x + \delta \neq 0 $. If $ \gamma \neq 0 $, then it also acts as
$$ \psi_M ( - \delta \gamma^{-1} ) = \infty \quad \textrm{and} \quad \psi_M(\infty) = \alpha \gamma^{-1}. $$
We say then that $ -\delta \gamma^{-1} \in \F_q $ is the pole of $ \psi_M $. Finally, we define the group of all projectivities in $ \mathbb{P}^1(\F_q) $ as
$$ \PGL(\F_q) = \{ \psi_M : M \in \GL_2(\F_q) \}. $$
\end{definition}

Observe that projectivities are bijective (i.e., permutations of $ \mathbb{P}^1(\F_q) $) and two matrices $ M,N \in \GL_2(\F_q) $ define the same projectivity ($ \psi_M=\psi_N $) if and only if there exists $ \lambda \in \F_q^* $ such that $ N = \lambda M $. Note also that $ \psi_M(\infty) = \infty $ (i.e., $ \psi_M(\F_q) = \F_q $) if and only if $ \gamma = 0 $, in which case we may take $ \delta = 1 $ and the projectivity becomes
$$ \psi_M (x) = \alpha x + \beta, $$ 
for $ x \in \F_q $, which is an affinity (or affine map) of the line $ \F_q $. 

Note also that, in all cases, since $ M \in \GL_2(\F_q) $ is invertible, then
$$ \Delta = \alpha \delta - \beta \gamma \neq 0 . $$

We will also need cross ratios.

\begin{definition} [\textbf{Cross ratios}]
Given distinct $ a,b,c,d \in \mathbb{P}^1(\F_q) $ with representatives $ u,v,w,t \in \F_q^2 \setminus \{ 0 \} $, respectively, we define the cross ratio
$$ [a,b;c,d] = \frac{\det(u^\intercal,w^\intercal) \det(v^\intercal,t^\intercal)}{\det(u^\intercal,t^\intercal) \det(v^\intercal,w^\intercal)} \in \F_q^* , $$
which does not depend on the chosen representatives. When $ a,b,c,d \in \F_q $, their cross ratio coincides with
$$ [a,b;c,d] = \frac{(a-c)(b-d)}{(a-d)(b-c)}. $$
\end{definition}

We will need the following classical facts, see \cite[Cor. 3.6.3]{BeutelspacherRosenbaum1998} and \cite[Sec. 6.1]{Hirschfeld1998}.

\begin{proposition} \label{prop projectivities 1}
Given three distinct points $ a,b,c \in \mathbb{P}^1(\F_q) $ and another three distinct points $ a^\prime, b^\prime, c^\prime \in \mathbb{P}^1(\F_q) $, there exists a unique projectivity $ \psi_M : \mathbb{P}^1(\F_q) \longrightarrow \mathbb{P}^1(\F_q) $ such that
$$ \psi_M(a) = a^\prime, \quad \psi_M(b) = b^\prime \quad \textrm{and} \quad \psi_M(c) = c^\prime. $$
\end{proposition}

\begin{proposition} \label{prop projectivities 2}
Projectivities preserve cross ratios. That is, given distinct $ a,b,c,d \in \mathbb{P}^1(\F_q) $ and a projectivity $ \psi_M : \mathbb{P}^1(\F_q) \longrightarrow \mathbb{P}^1(\F_q) $, then
$$ [\psi_M(a),\psi_M(b);\psi_M(c),\psi_M(d)] = [a,b;c,d]. $$
\end{proposition}

\begin{proposition} \label{prop projectivities 3}
Fix three distinct $ a,b,c \in \F_q $. Then the map $ \phi : \F_q \setminus \{ a,b,c \} \longrightarrow \F_q $ given by $ \phi(d) = [a,b;c,d] $, for all $ d \in \F_q \setminus \{a,b,c \} $, is injective.
\end{proposition}

\section{Projectivities yield NRT automorphisms} \label{sec projectivities autom}

In this section, we start by showing that the projectivities that keep $ \A $ invariant are all automorphisms of $ \M_\A(r,k) $. 
%
%
Throughout this section, fix an arbitrary
$$ M = \left( \begin{array}{cc}
\alpha & \gamma \\
\beta & \delta
\end{array} \right) \in \GL_2(\F_q). $$
On the polynomial side, projectivities will act as follows.

\begin{definition} \label{def Psi M k}
We define $ \Psi_{M,k} : \V_k \longrightarrow \V_k $ by
$$ \Psi_{M,k}(f) = (\gamma x + \delta)^{k-1} f \left( \frac{\alpha x + \beta}{\gamma x + \delta} \right). $$
Sometimes we will write $ \Psi_{M,k} f $ instead of $ \Psi_{M,k}(f) $.
\end{definition}

\begin{proposition} \label{prop psi Mk is automorph of Vk}
The map $ \Psi_{M,k} $ from Definition \ref{def Psi M k} is well defined and a vector space automorphism of $ \V_k $.
\end{proposition}
\begin{proof}
If $ f = f_0 + f_1 x + \cdots + f_{k-1} x^{k-1} \in \V_k $, then
$$ \Psi_{M,k}(f) = \sum_{i=0}^{k-1} f_i (\alpha x + \beta)^i (\gamma x + \delta)^{k-1-i} \in \V_k. $$
Thus clearly $ \Psi_{M,k} $ is well-defined, linear and its image is contained in $ \V_k $. Now, for another 
$$ N = \left( \begin{array}{cc}
a & c \\
b & d
\end{array} \right) \in \GL_2(\F_q) , $$
we have that
\begin{equation*}
\begin{split}
\Psi_{M,k}(\Psi_{N,k}(f)) & = (\gamma x + \delta)^{k-1} (c \psi_M(x) + d )^{k-1} f(\psi_N(\psi_M(x))) \\
 & = ((c\alpha + d \gamma) x + (c \beta + d \delta))^{k-1} f(\psi_{MN}(x)) = \Psi_{MN,k}(f).
\end{split}
\end{equation*}
In other words, $ \Psi_{M,k} \circ \Psi_{N,k} = \Psi_{MN,k} $, which implies that $ \Psi_{M,k} $ is invertible with inverse $ \Psi_{M^{-1},k} $ (since the identity matrix yields the identity automorphism of $ \V_k $), and we are done.
\end{proof}

We may now show that projectivities give NRT automorphisms of $ \M_\A(r,k) $. The key point is that the associated matrices in each block are upper triangular. 

\begin{proposition} \label{prop projectivities are automorph}
Let $ a \in \F_q $ be a point distinct from the pole of $ \psi_M $. Set
$$ b = \psi_M(a) \in \F_q \quad \textrm{and} \quad c_a = \gamma a + \delta \in \F_q^*. $$
For $ m \in [k] $, define the upper triangular matrix $ R^{(m)}_{a,M,k} \in \F_q^{m \times m} $ by
$$ (R^{(m)}_{a,M,k})_{i,j} = \begin{cases}
 \displaystyle
 \binom{k-1-i}{j-i}
 \Delta^i\gamma^{j-i}c_a^{\,k-1-i-j},& \textrm{ if } i\leq j,\\
 0,& \textrm{ if } i>j,
 \end{cases} $$
for $ i,j \in [0,m-1] $. Then $ R^{(m)}_{a,M,k} $ is invertible and 
\begin{equation}
J^{(m)}_a(\Psi_{M,k}(f)) = J^{(m)}_b(f) \cdot R^{(m)}_{a,M,k},
\label{eq jets Psi and R}
\end{equation}
for all $ f \in \V_k $.
\end{proposition}
\begin{proof}
First, $ R^{(m)}_{a,M,k} $ is invertible since it is clearly upper triangular and the diagonal elements are, for all $ i \in [0,m-1] $,
$$ (R^{(m)}_{a,M,k})_{i,i} = \Delta^i c_a^{k-1-2i} \neq 0, $$
where $ \Delta = \det(M) \neq 0 $, because $ M \in \GL_2(\F_q) $, and $ c_a = \gamma a + \delta \neq 0 $ since $ a $ is not a pole of $ \psi_M $.

Now let $ y $ be an algebraic variable and take $ x = a+y $. Multiplying by denominators, the reader can verify that
$$ \psi_M(a+y) = b + \frac{\Delta y}{c_a(c_a+ \gamma y)}. $$
Therefore, by Definition \ref{def hasse derivative}, we have that
\begin{equation*}
\begin{split}
\Psi_{M,k}f (a+y) & = (c_a + \gamma y)^{k-1} f \left( b + \frac{\Delta y}{c_a(c_a+\gamma y)} \right) \\
 & = (c_a + \gamma y)^{k-1} \sum_{i \in \N} f^{(i)}(b) \left( \frac{\Delta y}{c_a(c_a+\gamma y)} \right)^i \\
 & = \sum_{i \in \N} f^{(i)}(b) \Delta^i c_a^{-i}y^i (c_a + \gamma y)^{k-1-i} \\
 & = \sum_{i \in \N} f^{(i)}(b) \Delta^i c_a^{-i} y^i \sum_{\ell=0}^{k-1-i} \binom{k-1-i}{\ell} \gamma^\ell c_a^{k-1-i-\ell} y^\ell \\
 & = \sum_{i \in \N} f^{(i)}(b) \sum_{j=i}^{k-1} \binom{k-1-i}{j-i} \Delta^i \gamma^{j-i} c_a^{k-1-i-j} y^j,
\end{split}
\end{equation*}
where in the last equality we define $ j = i + \ell $. Thus, the coefficient of $ y^j $ in the $ i $th term of the sum is zero if $ i > j $, and adding the coefficients of $ y^j $ for $ i \in [0,j] $, we obtain
$$ (\Psi_{M,k}f)^{(j)}(a) = \sum_{i=0}^j f^{(i)}(b) \binom{k-1-i}{j-i} \Delta^i \gamma^{j-i} c_a^{k-1-i-j}, $$
for $ j \in [0,k-1] $, which proves (\ref{eq jets Psi and R}), and we are done.
\end{proof}

\begin{remark} \label{remark affine diagonal matrix}
In the affine case, i.e., if we choose $ \gamma = 0 $ and the representative
$$ M = \left( \begin{array}{cc}
\alpha & 0 \\
\beta & 1 
\end{array} \right) \in \GL_2(\F_q), $$
then $ \psi_M(x) = \alpha x + \beta $, for $ x \in \F_q $, and
$$ R^{(m)}_{a,M,k} = \diag(1,\alpha,\alpha^2, \ldots, \alpha^{m-1}) \in \GL_m(\F_q). $$
Notice that it is independent of the value of $ a $ and $ k $ in this affine case, but not in general.
\end{remark}

\begin{example}
Let $ \A = \F_q^* $ and
$$ M = \left( \begin{array}{cc}
0 & 1 \\
1 & 0 
\end{array} \right) \in \GL_2(\F_q) , $$
i.e., the ``inversion'' projectivity $ \psi_M(x) = x^{-1} $, for $ x \in \A $. Then $ \psi_M (\A) = \A $ and
$$ \psi_{M,k}(f) = x^{k-1} f(x^{-1}), $$
for all $ f \in \V_k $. In particular, for $ a \in \F_q^* $, $ k \geq 3 $ and $ m = 2,3 $, we obtain
\[
R^{(2)}_{a,M,k}
=
\begin{pmatrix}
a^{k-1} & (k-1)a^{k-2}\\
0 & -a^{k-3}
\end{pmatrix} \quad \textrm{and} \quad R^{(3)}_
{a,M,k}
=
\begin{pmatrix}
a^{k-1}
&
(k-1)a^{k-2}
&
\binom{k-1}{2}a^{k-3}
\\[1mm]
0
&
-a^{k-3}
&
-(k-2)a^{k-4}
\\[1mm]
0
&
0
&
a^{k-5}
\end{pmatrix}.
\]
\end{example}

As an immediate consequence of Proposition \ref{prop projectivities are automorph}, we obtain the following family of automorphisms of multiplicity codes.

\begin{corollary} \label{cor T_M given by projectivity}
Assume that $ \psi_M(\A) = \A $ (in particular, $ a $ is not the pole of $ \psi_M $, for all $ a \in \A $). If $ k \geq r $, then $ T_M : (\F_q^r)^n \longrightarrow (\F_q^r)^n $ given by
$$ T_M((z_a)_{a \in \A}) = (z_{\psi_M(a)} R^{(r)}_{a,M,k})_{a \in \A}, $$
for all $ (z_a)_{a \in \A} \in (\F_q^r)^n $, belongs to $ \Aut(\M_\A(r,k)) $. If $ k < r $, the same holds for $ T_M : (\F_q^r)^n \longrightarrow (\F_q^r)^n $ given by
$$ T_M((z_a)_{a \in \A}) = (z_{\psi_M(a)} \widetilde{R}_{a,M,k})_{a \in \A}, $$
where
$$ \widetilde{R}_{a,M,k} = \left( \begin{array}{cc}
R^{(k)}_{a,M,k} & 0 \\
0 & I_{r-k}
\end{array} \right) \in \GL_r(\F_q) . $$
\end{corollary}
\begin{proof}
Combine Equation (\ref{eq jets Psi and R}) with Proposition \ref{prop psi Mk is automorph of Vk}. In the case $ k < r $, use Equation (\ref{eq jets Psi and R}) with $ m=k $ and notice that the last $ r-k $ components of $ J^{(r)}_a(f) $ are zero for all $ f \in \V_k $ and $ a \in \A $.
\end{proof}

In the remainder of the paper, we show that for most dimensions these are all the elements in $ \Aut(\M_\A(r,k)) $, except for multiplication by nonzero scalars.

\section{The permutation is a projectivity} \label{sec permutation is projectivity}

In this section, fix $ T \in \Aut(\M_\A(r,k)) $. By Proposition \ref{prop isometries NRT charact}, there exist a permutation $ \pi : \A \longrightarrow \A $ and upper triangular matrices $ A_a \in \GL_r(\F_q) $, for $ a \in \A $, such that
$$ T((z_a)_{a \in \A}) = (z_{\pi(a)} A_a)_{a \in \A}, $$
for all $ (z_a)_{a \in \A} \in (\F_q^r)^n $. We will need the following linear operator, which will end up being $ \Psi_{M,k} $ from Proposition \ref{prop psi Mk is automorph of Vk}, up to a nonzero scalar.

\begin{proposition} \label{prop L linear op}
There exists a unique vector space automorphism $ L : \V_k \longrightarrow \V_k $ such that
\begin{equation}
J^{(r)}_a(Lf) = J^{(r)}_{\pi(a)}(f) A_a,
\label{eq def L}
\end{equation}
for all $ a \in \A $ and all $ f \in \V_k $. We will indistinctively write $ L(f) $ or $ Lf $, for $ f \in \V_k $.
\end{proposition}
\begin{proof}
Existence follows because $ T \in \Aut(\M_\A(r,k)) $ and uniqueness by Proposition \ref{prop ev is injective}.
\end{proof}

We will also need the following multiplicity-root subspaces of $ \V_k $.

\begin{definition} \label{def K_a}
For $ a \in \A $ and $ t \in [0, \min\{r,k \}] $, we define the subspace
$$ \K_{a,t} = (x-a)^t \V_{k-t} \subseteq \V_k, $$
or in other words, the subspace of polynomials of degree less than $ k $ that have $ a $ as a root of multiplicity at least $ t $ (see Proposition \ref{prop multiplicity hasse derivatives}).
\end{definition}

Next we show that $ L $ sends such subspaces in the ``opposite direction'' than the permutation $ \pi $ associated to $ T $. Notice that here we use that the matrices $ A_a $ are upper triangular.

\begin{lemma} \label{lemma L(K) is K pi -1}
For every $ a \in \A $ and $ t \in [0, \min\{ r,k\}] $, 
$$ L(\K_{\pi(a),t}) = \K_{a,t}. $$
\end{lemma}
\begin{proof}
Let $ c = (c_0,c_1, \ldots, c_{r-1}) \in \F_q^r $. Since $ A_a $ is upper triangular and invertible, then the first $ t $ components of $ c $ are zero if and only if so are the first $ t $ components of $ c A_a $ (i.e., $ \wt(cA_a) = \wt(c) $ for NRT weights, see also Proposition \ref{prop isometries NRT charact}). Applying this fact to the defining Equation (\ref{eq def L}) and using Proposition \ref{prop multiplicity hasse derivatives}, the result follows.
\end{proof}

We now consider subspaces for multiplicities of roots among all elements of $ \A $.

\begin{definition}
For $ \ee = (e_a)_{a \in \A} \in [0,r]^\A $, define
$$ |\ee| = \sum_{a \in \A} e_a, \quad P_\ee(x) = \prod_{a \in \A} (x-a)^{e_a} \quad \textrm{and if } |\ee| < k, \textrm{ define} \quad \K_\ee = \bigcap_{a \in \A} \K_{a,e_a}. $$
For $ b \in \A $, we define $ \ones_b \in [0,r]^\A $ as the indicator function of $ b $, i.e., for all $ a \in \A $,
$$ (\ones_b)_a = \left\lbrace \begin{array}{ll}
0 & , \textrm{ if } a \neq b, \\
1 & , \textrm{ if } a = b.
\end{array} \right. $$
\end{definition}

The proof of the following result is straightforward using Lemma \ref{lemma L(K) is K pi -1}.

\begin{lemma} \label{lemma K_ee is P_ee V k - ee}
If $ \ee \in [0,r]^\A $ and $ |\ee| < k $, then
$$ \K_\ee = P_\ee \V_{k-|\ee|}. $$
Moreover, if $ \rho = \pi^{-1} : \A \longrightarrow \A $, and we define $ \rho(\ee) \in [0,r]^\A $ by
$$ \rho(\ee)_a = e_{\pi(a)}, $$
for all $ a \in \A $, then we have
$$ L(\K_\ee) = \K_{\rho(\ee)}. $$
In particular, if $ |\ee| = k-1 $, then $ L(P_\ee) = \lambda_\ee P_{\rho(\ee)} $, for some $ \lambda_\ee \in \F_q^* $. 
\end{lemma}

As we will show later, the vector lines $ \K_\ee $, for $ |\ee|=k-1 $, are crucial for showing that $ \pi $ is given by a projectivity. However, in order to relate these lines, we will need to use the vector planes $ \K_\ee $, for $ |\ee| = k-2 $. Subtract one dimension for the projective point of view. 

\begin{lemma} \label{lemma ee with prescribed values}
Assume $ n = |\A| \geq 4 $ and $ 2 \leq k \leq rn - r $. Assume also $ k \leq n-2 $ if $ r=1 $. Given four distinct points $ a,b,c,d \in \A $, there exists $ \ee \in [0,r]^\A $ such that
$$ |\ee| = k-2 \quad \textrm{and} \quad \max\{ e_a,e_b,e_c,e_d \} \leq r-1. $$
\end{lemma}
\begin{proof}
If we allow $ \max\{ e_a,e_b,e_c,e_d \} \leq r-1 $ and $ e_z \leq r $, for all $ z \in \A \setminus \{ a,b,c,d \} $, then $ \ee $ may have total sum up to
$$ |\ee| \leq 4(r-1) + (n-4)r = nr-4. $$
If $ r \geq 2 $, since $ k \leq rn-r $, then
$$ k-2 \leq rn - r -2 \leq rn - 4. $$
In the case $ r=1 $, since $ k \leq n-2 $, then we also have $ k - 2 \leq n-4 = rn -4 $. Thus $ k-2 $ multiplicities may be distributed among the components of $ \ee $ as desired.
\end{proof}

We may now start to see why $ \pi $ is a projectivity: it preserves cross ratios.

\begin{lemma} \label{lemma pi -1 preserves cross ratios}
Assume $ n = |\A| \geq 4 $ and $ 2 \leq k \leq rn - r $. Assume also $ k \leq n-2 $ if $ r=1 $. If $ a,b,c,d \in \A $ are distinct, then
$$ [\pi(a),\pi(b); \pi(c),\pi(d)] = [a,b;c,d]. $$
\end{lemma}
\begin{proof}
Choose $ \ee \in [0,r]^\A $ as in Lemma \ref{lemma ee with prescribed values}. Since $ |\ee| = k-2 $, by Lemma \ref{lemma K_ee is P_ee V k - ee}, we have that
$$ \K_\ee = P_\ee \V_2 = \langle x P_\ee, P_\ee \rangle_{\F_q}, $$
which has dimension $ 2 $. For $ z \in \{ a,b,c,d \} $, since $ e_z \leq r-1 $, we may increase the multiplicity at $ z $ by one, obtaining the subspace
$$ \K_{\ee+\ones_z} = \langle (x-z) P_\ee \rangle_{\F_q} \subseteq \K_\ee, $$
which has dimension $ 1 $. By checking each component at every $ y \in \A $, the reader may verify that
$$ \rho(\ee + \ones_z) = \rho(\ee) + \ones_{\rho(z)}, $$
for all $ z \in \{ a,b,c,d \} $, for the permutation $ \rho = \pi^{-1} : \A \longrightarrow \A $. Thus by Lemma \ref{lemma K_ee is P_ee V k - ee}, 
$$ L(\K_{\ee+\ones_z}) = \K_{\rho(\ee) + \ones_{\rho(z)}}. $$ 
In the ordered bases $ (xP_\ee, P_\ee) $ and $ (xP_{\rho(\ee)}, P_{\rho(\ee)}) $ of $ \K_\ee $ and $ \K_{\rho(\ee)} $, respectively, the lines $ \K_{\ee + \ones_z} $ and $ \K_{\rho(\ee) + \ones_{\rho(z)}} $ are generated by the vectors with coordinates $ (1,-z), (1,-\rho(z)) \in \F_q^2 $, respectively, for $ z \in \{ a,b,c,d \} $. Hence the projectivity associated to $ L : \K_\ee \longrightarrow \K_{\rho(\ee)} $, in such bases, sends the projective point $ [1:-z] $ to $ [1:-\rho(z)] $, for $ z \in \{ a,b,c,d \} $. Since projectivities preserve cross ratios (Proposition \ref{prop projectivities 2}) and cross ratios do not depend on the chosen bases (which is essentially the same result), we deduce that
\begin{equation*}
\begin{split}
[\rho(a),\rho(b);\rho(c),\rho(d)] & = [[1: -\rho(a)],[1: -\rho(b)];[1: -\rho(c)], [1: -\rho(d)] ] \\
 & = [[1: -a],[1: -b];[1: -c], [1: -d] ] \\
 & = [a,b;c,d].
\end{split}
\end{equation*}
Since $ \rho = \pi^{-1} $, we deduce that $ \pi $ also preserves cross ratios.
\end{proof}

Finally, we can show that $ \pi $ is given by a projectivity.

\begin{proposition} \label{prop pi is a projectivity}
Assume $ n = |\A| \geq 3 $ and $ 2 \leq k \leq rn - r $. Assume also $ k \leq n-2 $ if $ r=1 $. Then there exists a unique projectivity $ \psi_M \in \PGL(\F_q) $ such that
$$ \pi(a) = \psi_M (a), $$
for all $ a \in \A $. In particular, $ \psi_M(\A) = \A $.
\end{proposition}
\begin{proof}
If $ n=3 $, any permutation of $ \A $ can be extended to a projectivity of $ \mathbb{P}^1(\F_q) $ by Proposition \ref{prop projectivities 1}. Now assume $ n \geq 4 $. Fix three distinct points $ a,b,c \in \A $. By Proposition \ref{prop projectivities 1}, there exists a unique projectivity $ \psi_M \in \PGL(\F_q) $ such that
$$ \psi_M(a) = \pi(a), \quad \psi_M(b) = \pi(b) \quad \textrm{and} \quad \psi_M(c) = \pi(c). $$
For any $ d \in \A \setminus \{ a,b,c \} $, Lemma \ref{lemma pi -1 preserves cross ratios} for $ \pi $ and Proposition \ref{prop projectivities 2} for $ \psi_M $ yield
$$ [\pi(a),\pi(b);\pi(c),\pi(d)] = [a,b;c,d] = [\pi(a),\pi(b); \pi(c), \psi_M(d)]. $$
By Proposition \ref{prop projectivities 3}, we have that $ \pi(d) = \psi_M(d) $, for all $ d \in \A \setminus \{a,b,c \} $. Thus $ \pi = \psi_M $ in $ \A $. 
\end{proof}

\begin{remark}
If $ \A = \F_q $, then the condition $ \psi_M (\A) = \A $ means that $ \psi_M(\infty) = \infty $. As shown after Definition \ref{def projectivity}, this means that $ \psi_M $ is affine. Hence for $ \A = \F_q $, Proposition \ref{prop pi is a projectivity} says that there exist $ \alpha, \beta \in \F_q $, with $ \alpha \neq 0 $ (since $ \alpha = \Delta $), such that
$$ \pi(a) = \alpha a + \beta, $$
for all $ a \in \F_q $.
\end{remark}

\section{The $ (k-1) $-degree graph and the scalar $ \lambda $} \label{sec d degree graph}

In this section, we assume that 
\begin{equation}
n = |\A| \geq 3 , \quad 2 \leq k \leq rn - r, \quad \textrm{and} \quad k \leq n-2 \textrm{ if } r=1 .
\label{eq assumption degree graph sec}
\end{equation}
Let $ T \in \Aut(\M_\A(r,k)) $. As shown in Proposition \ref{prop isometries NRT charact}, $ T $ consists of a permutation $ \pi : \A \longrightarrow \A $ and products by upper triangular invertible matrices. By Proposition \ref{prop pi is a projectivity}, there exists a projectivity $ \psi_M \in \PGL(\F_q) $ such that $ \pi(a) = \psi_M(a) $, for all $ a \in \A $. In particular, since $ T_M \in \Aut(\M_\A(r,k)) $ also has permutation associated to $ \psi_M $ by Corollary \ref{cor T_M given by projectivity}, then $ T \circ T_M^{-1} \in \Aut(\M_\A(r,k)) $ is an automorphism whose permutation is the identity.

Thus, in this section we just fix $ T \in \Aut(\M_\A(r,k)) $ whose permutation $ \pi $ is the identity in $ \A $. We will show in this section that then there exists $ \lambda \in \F_q^* $ such that $ L(f) = \lambda f $, for all $ f \in \V_k $. This will later be used to show that $ T(c) = \lambda c $, for all $ c \in (\F_q^r)^n $, which will conclude the classification problem for most parameters.

By Lemma \ref{lemma K_ee is P_ee V k - ee}, there exists $ \lambda_\ee \in \F_q^* $ such that
\begin{equation}
L(P_\ee) = \lambda_\ee P_\ee,
\label{eq eigenpolynomials}
\end{equation}
for all $ \ee \in [0,r]^\A $ such that $ |\ee| = k-1 $. We will show that $ \lambda_\ee $ is independent of such $ \ee \in [0,r]^\A $ using the $ (k-1) $-degree graph.

\begin{definition} \label{def k-1 degree graph}
Let $ \mathcal{G} $ be the graph whose vertices are the vectors $ \ee \in [0,r]^\A $ such that $ |\ee| = k-1 $, and whose edges are the following. Given two vertices $ \ee $ and $ \ee^\prime $, $ \{ \ee,\ee^\prime \} $ is an edge of $ \mathcal{G} $ if
$$ \ee = \cc + \ones_a, \quad \ee^\prime = \cc + \ones_b, $$
for distinct $ a,b \in \A $ and some $ \cc \in [0,r]^\A $ (with $ |\cc|=k-2 $) with at least three components strictly smaller than $ r $ (notice two of them correspond to $ a $ and $ b $).
\end{definition}

The key fact is that this graph is connected.

\begin{lemma} \label{lemma graph is connected}
Assuming (\ref{eq assumption degree graph sec}), the $ (k-1) $-degree graph $ \mathcal{G} $ from Definition \ref{def k-1 degree graph} is connected.
\end{lemma}
\begin{proof}
Clearly, $ \mathcal{G} $ is a subgraph of the ordinary transfer graph $ \mathcal{G}^\prime $. Recall that $ \mathcal{G}^\prime $ is the graph whose vertices are the same vectors $ \ee \in [0,r]^\A $, with $ |\ee| = k-1 $, and where two such vertices $ \ee,\ee^\prime $ form an edge if and only if one can obtain $ \ee^\prime $ from $ \ee $ by subtracting $ 1 $ in one component and adding it to another component. In other words, if there exists $ \cc \in [0,r]^\A $ (with $ |\cc| = k-2 $) and distinct $ a,b \in \A $ such that 
$$ \ee = \cc + \ones_a \quad \textrm{and} \quad \ee^\prime = \cc + \ones_b. $$
This graph $ \mathcal{G}^\prime $ is connected: given any two vertices $ \ee $ and $ \ee^\prime $, one can travel from $ \ee $ to $ \ee^\prime $ by successively transferring one unit from one component to another, since $ |\ee| = k-1 = | \ee^\prime| $. 

Consider now two vertices $ \ee, \ee^\prime $ of $ \mathcal{G} $ that are neighbours in $ \mathcal{G}^\prime $, i.e., there exists $ \cc \in [0,r]^\A $ and distinct $ a,b \in \A $ such that
$$ \ee = \cc + \ones_a \quad \textrm{and} \quad \ee^\prime = \cc + \ones_b. $$
If at least three components of $ \cc $ are strictly smaller than $ r $, then this edge is in $ \mathcal{G} $. Assume this is not the case, that is, $ a $ and $ b $ are the only coordinates with $ c_a \leq r-1 $ and $ c_b \leq r-1 $. Since $ |\cc| = k-2 $, then 
$$ rn - |\cc| = rn - k + 2 \geq r + 2, $$
by the assumptions at the beginning of this section. If we denote the deficits as $ \delta_a = r - c_a $ and $ \delta_b = r - c_b $, then
\begin{equation}
\delta_a + \delta_b \geq r+2 , \quad \textrm{hence} \quad c_a + c_b \leq r-2.
\label{eq c_a + c_b leq r-2}
\end{equation}
Notice that, if $ r=1 $, then $ \delta_a + \delta_b \leq 2 $, which contradicts the first inequality in (\ref{eq c_a + c_b leq r-2}). Hence if $ r=1 $, the edge must be in $ \mathcal{G} $ necessarily. 

From now on, we assume $ r \geq 2 $. Choose $ s \in \A \setminus \{ a,b \} $, which exists since we are assuming that $ n = |\A| \geq 3 $. Replace now the edge $ \{ \ee, \ee^\prime \} $ in $ \mathcal{G}^\prime $ (recall we are assuming it is not in $ \mathcal{G} $) by the two edges $ \{ \ee, \mm \} $ and $ \{ \mm, \ee^\prime \} $, also in $ \mathcal{G}^\prime $, where
$$ \mm = \cc - \ones_s + \ones_a + \ones_b. $$
Notice that $ \mm \in [0,r]^\A $ since we are assuming that $ c_s = r \geq 2 $, $ c_a \leq r-1 $ and $ c_b \leq r-1 $. For the edge $ \{ \ee, \mm \} $, we have that
$$ \ee = (\cc - \ones_s + \ones_a) + \ones_s \quad \textrm{and} \quad \mm = (\cc - \ones_s + \ones_a) + \ones_b, $$
where $ \cc - \ones_s + \ones_a \in [0,r]^\A $ has components strictly smaller than $ r $ in the three distinct coordinates $ a,b,s $, since we are assuming that $ c_s \leq r $, $ c_b \leq r-1 $ and it holds that $ c_a + 1 \leq r-1 $ by Equation (\ref{eq c_a + c_b leq r-2}). Therefore $ \{ \ee, \mm \} $ is an edge of $ \mathcal{G} $. In the same way, $ \{ \mm, \ee^\prime \} $ is also an edge of $ \mathcal{G} $. 

Hence, for any two vertices $ \ee , \ee^\prime $ of $ \mathcal{G} $, one can take a path between them inside the transfer graph $ \mathcal{G}^\prime $ (recall that it is connected), and replace each edge of the path, if necessary, by two edges in $ \mathcal{G} $ as shown above. We thus conclude that $ \mathcal{G} $ is connected.
\end{proof}

We are now ready to show that $ \lambda_\ee $ does not depend on $ \ee $.

\begin{lemma} \label{lemma one eignevalue}
Assuming (\ref{eq assumption degree graph sec}), there exists $ \lambda \in \F_q^* $ such that 
$$ L(P_\ee) = \lambda P_\ee, $$
for all $ \ee \in [0,r]^\A $ such that $ |\ee| = k-1 $.
\end{lemma}
\begin{proof}
Let $ \ee, \ee^\prime $ be two neighbour vertices in the $ (k-1) $-degree graph $ \mathcal{G} $. By definition, there exists $ \cc \in [0,r]^\A $ with $ |\cc| = k-2 $, with at least three components strictly smaller than $ r $ and 
$$ \ee = \cc + \ones_a \quad \textrm{and} \quad \ee^\prime = \cc + \ones_b, $$
for distinct $ a,b \in \A $. In particular, there exists $ s \in \A \setminus \{ a,b \} $ such that $ \max \{ c_a,c_b,c_s \} \leq r - 1 $. The $ 2 $-dimensional subspace
$$ K_\cc = P_\cc \V_2 $$
is invariant by $ L $ (by Lemma \ref{lemma K_ee is P_ee V k - ee} since $ \pi $ is the identity) and contains the three eigenvectors
$$ P_{\cc + \ones_a} = (x-a)P_\cc, \quad P_{\cc + \ones_b} = (x-b) P_\cc \quad \textrm{and} \quad P_{\cc + \ones_s} = (x-s)P_\cc, $$
by Equation (\ref{eq eigenpolynomials}). Since the three eigenvectors lie in a plane and any two of them are linearly independent, we conclude that their three eigenvalues coincide (the proof of this Linear Algebra fact is left to the reader). Now, the eigenvalues for $ P_{\cc + \ones_a} = P_\ee $ and $ P_{\cc + \ones_b} = P_{\ee^\prime} $ are, respectively, $ \lambda_\ee $ and $ \lambda_{\ee^\prime} $. Thus they are equal. 

In other words, the eigenvalues $ \lambda_\ee $ are equal for neighbour vertices in $ \mathcal{G} $. Since this graph is connected by Lemma \ref{lemma graph is connected}, we conclude that the eigenvalues in Equation (\ref{eq eigenpolynomials}) are equal for all $ \ee \in [0,r]^\A $ such that $ |\ee| = k-1 $, and we are done.
\end{proof}

A final key ingredient is that $ P_\ee $, for $ |\ee|=k-1 $, generate the whole space $ \V_k $.

\begin{lemma} \label{lemma Pe generated Vk}
Assuming (\ref{eq assumption degree graph sec}), For all $ j \in [0,k-1] $, we have that
$$ \langle \{ P_\ee : \ee \in [0,r]^\A, |\ee| = j \} \rangle_{\F_q} = \F_q[x]_{\leq j} = \V_{j+1}. $$
\end{lemma}
\begin{proof}
We proceed by induction on $ j $. For $ j=0 $, the result is trivial. Assume the lemma holds for $ j \in [0,k-2] $, and let us prove it for $ j+1 $. Let $ \cc \in [0,r]^\A $ be such that $ |\cc| = j $. At least two components of $ \cc $ are strictly smaller than $ r $. Otherwise, at least $ n-1 $ of its components would equal $ r $, giving
$$ j \geq (n-1)r = nr- r, $$
and we are assuming that
$$ j \leq k-2 \leq nr - r - 2 = nr - r -2. $$
Hence we may choose distinct $ a,b \in \A $ such that $ c_a < r $ and $ c_b < r $. Then
$$ P_{\cc + \ones_a} - P_{\cc + \ones_b} = (b-a) P_\cc. $$
Thus for all $ \cc \in [0,r]^\A $ such that $ |\cc| = j $, we have 
$$ P_\cc \in \langle \{ P_\ee : \ee \in [0,r]^\A, |\ee| = j+1 \} \rangle_{\F_q} . $$
By induction hypothesis, we have that 
$$ \F_q[x]_{\leq j} \subseteq \langle \{ P_\ee : \ee \in [0,r]^\A, |\ee| = j+1 \} \rangle_{\F_q}. $$
Since there exists at least one $ \ee \in [0,r]^\A $ with $ |\ee| = j+1 $, and $ P_\ee $ is monic of degree $ j+1 $, we conclude the result.
\end{proof}

We may now conclude the main result of this section. Recall that, at the beginning of this section, we fixed $ T \in \Aut(\M_\A(r,k)) $ whose associated permutation $ \pi : \A \longrightarrow \A $ is the identity in $ \A $. Recall the map $ L : \V_k \longrightarrow \V_k $ from Proposition \ref{prop L linear op}.

\begin{theorem} \label{th automorph with pi = id is lambda}
Assuming (\ref{eq assumption degree graph sec}), there exists $ \lambda \in \F_q^* $ such that, for all $ f \in \V_k $,
$$ L(f) = \lambda f. $$
\end{theorem}
\begin{proof}
By Lemma \ref{lemma one eignevalue}, there exists $ \lambda \in \F_q^* $ such that $ L(P_\ee) = \lambda P_\ee $, for all $ \ee \in [0,r]^\A $ such that $ |\ee| = k-1 $. Since these polynomials generate $ \V_k $ by Lemma \ref{lemma Pe generated Vk}, we conclude that $ L(f) = \lambda f $, for all $ f \in \V_k $. 
\end{proof}

\section{The automorphism group} \label{sec classification}
 
In this section, we collect the results from Sections \ref{sec permutation is projectivity} and \ref{sec d degree graph} to show that automorphisms of multiplicity codes are given by projectivities, up to multiplication by a scalar in $ \F_q^* $, for the ``central dimension'' range, meaning $ n \geq 3 $ and $ r \leq k \leq nr - r $ (and $ 2 \leq k \leq n-2 $ when $ r=1 $). 

The ``low and high'' dimensional ranges ($ 2 \leq k \leq r-1 $ and $ nr-r+1 \leq k \leq nr-2 $) are treated in Appendix \ref{app low high range}. The rest of the cases (i.e., the four cases $ k=1 $, $ k=rn-1 $, $ n=1 $ and $ n= 2 $) are treated in Appendix \ref{app exceptional cases}. Notice that for $ k = rn $, we obtain the full-space code $ (\F_q^r)^n $.

The following result recovers (by setting $ r=1 $) the classical characterization of the automorphism group (in the classical Hamming metric) of Reed--Solomon codes by Dur \cite[Th. 4]{Dur1987}.

\begin{theorem} [\textbf{Central range}] \label{th NRT aut central range}
Assume that $ n = |\A| \geq 3 $ and $ r \leq k \leq nr- r $. Assume also $ 2 \leq k \leq n-2 $ if $ r=1 $. A linear isometry $ T : (\F_q^r)^n \longrightarrow (\F_q^r)^n $ for the NRT metric lies in $ \Aut(\M_\A(r,k)) $ if and only if there exists $ \lambda \in \F_q^* $ and a projectivity $ \psi_M \in \PGL(\F_q) $ such that $ \psi_M(\A) = \A $ and
$$ T((z_a)_{a \in \A}) = \lambda \left(z_{\psi_M(a)} R^{(r)}_{a,M,k} \right)_{a \in \A}, $$
for all $ (z_a)_{a \in \A} \in (\F_q^r)^n $, where the invertible upper triangular matrices $ R^{(r)}_{a,M,k} \in \F_q^{r \times r} $ are given in Proposition \ref{prop projectivities are automorph}. On polynomials, we have that $ L : \V_k \longrightarrow \V_k $ (Proposition \ref{prop L linear op}) acts as 
$$ L(f) = \lambda \Psi_{M,k}(f) = \lambda (\gamma x + \delta)^{k-1} f \left( \frac{\alpha x + \beta}{\gamma x + \delta} \right), $$
for all $ f \in \V_k $, where the matrix associated to $ \psi_M $ is
$$ M = \left( \begin{array}{cc}
\alpha & \gamma \\
\beta & \delta
\end{array} \right) \in \GL_2(\F_q) . $$
Finally, the projectivity $ \psi_M $ is unique.
\end{theorem}
\begin{proof}
Any linear isometry of the form given in the theorem belongs to $ \Aut(\M_\A(r,k)) $ by Corollary \ref{cor T_M given by projectivity}. Conversely, given $ T \in \Aut(\M_\A(r,k)) $, by Proposition \ref{prop pi is a projectivity}, the permutation associated to $ T $ is given by a unique projectivity $ \psi_M \in \PGL(\F_q) $ such that $ \psi_M(\A) = \A $. Now, consider $ T^\prime = T \circ T_M^{-1} \in \Aut(\M_\A(r,k)) $, as at the start of Section \ref{sec d degree graph} (see Corollary \ref{cor T_M given by projectivity} for the definition of $ T_M $). Hence the permutation associated to $ T^\prime $ is the identity. Let $ L^\prime $ be the automorphism of $ \V_k $ associated to $ T^\prime $ from Proposition \ref{prop L linear op}, and let $ A^\prime_a \in \GL_r(\F_q) $, for $ a \in \A $, be the upper triangular matrices associated to $ T^\prime $ from Proposition \ref{prop isometries NRT charact}. By Theorem \ref{th automorph with pi = id is lambda}, there exists $ \lambda \in \F_q^* $ such that
$$ L^\prime(f) = \lambda f, $$ 
for all $ f \in \V_k $. Combined with Proposition \ref{prop L linear op} and the fact that the permutation associated to $ T^\prime $ is the identity, we deduce that
$$ \lambda J^{(r)}_a(f) = J^{(r)}_a(L^\prime f) = J^{(r)}_a(f) A^\prime_a, $$
for all $ a \in \A $. Now since $ k \geq r $, then the evaluation map 
$$ {\rm ev}_a : \V_k \longrightarrow \F_q^r, $$
given by $ {\rm ev}(f) = J^{(r)}_a(f) $, for $ f \in \V_k $, is surjective. We deduce that 
$$ A^\prime_a = \diag(\lambda, \lambda, \ldots, \lambda) = \lambda I_r \in \F_q^{r \times r} , $$
for all $ a \in \A $. In other words, $ T^\prime (c) = \lambda c $, for all $ c \in (\F_q^r)^n $. Thus $ T = \lambda T_M $, as desired. 
\end{proof}

As a group, we have the following isomorphism.

\begin{corollary}
Assume that $ n = |\A| \geq 3 $ and $ r \leq k \leq nr- r $. Assume also $ 2 \leq k \leq n-2 $ if $ r=1 $. Define $ \GL(\A) $ as the subgroup of matrices $ M \in \GL_2(\F_q) $ such that $ \psi_M(\A) = \A $. We have a natural group isomorphism
$$ \Aut(\M_\A(r,k)) \cong (\F_q^* \times \GL(\A)) / \mathcal{H}_k, $$
for the subgroup $ \mathcal{H}_k = \{ (c^{-(k-1)},cI_2) : c \in \F_q^* \} $ of $ \F_q^* \times \GL(\A) $. The isomorphism sends $ T \in \Aut(\M_\A(r,k)) $ to the class of $ (\lambda,M) $ modulo $ \mathcal{H}_k $, where $ \lambda \in \F_q^* $ and $ M \in \GL(\A) $ are as in Theorem \ref{th NRT aut central range}. 
\end{corollary}
\begin{proof}
Define the group morphism
$$ \phi : \F_q^* \times \GL(\A) \longrightarrow \Aut(\M_\A(r,k)) $$
by $ \phi(\lambda,M) = \lambda T_M $ (Corollary \ref{cor T_M given by projectivity}), which is surjective by Theorem \ref{th NRT aut central range}. Now, let $ (\lambda,M), (\mu,N) \in \F_q^* \times \GL(\A) $ be such that $ \lambda T_M = \mu T_N $. By uniqueness of the projectivity, we have that $ \psi_M = \psi_N $, thus there exists $ c \in \F_q^* $ such that $ N = cM $. Now, we have that 
$$ R^{(r)}_{a,cM,k} = c^{k-1} R^{(r)}_{a,M,k}, $$
for $ a \in \A $. Hence
$$ \lambda T_M = \mu T_N = \mu T_{cM} = \mu c^{k-1} T_M. $$
This implies that $ \mu = c^{-(k-1)} \lambda $. In particular, $ \mu T_N $ is the identity if and only if $ N = c I_2 $ and $ \mu = c^{-(k-1)} $, thus $ \ker(\phi) = \mathcal{H}_k $ and the result follows.
\end{proof}

Define $ \PGL(\A) $ as the subgroup of projectivities $ \psi_M \in \PGL(\F_q) $ such that $ \psi_M(\A) = \A $. Although it is true that $ \PGL(\A) \cong \GL(\A) / \{ c I_2 : c \in \F_q^* \} $, in general we have that
$$ (\F_q^* \times \GL(\A)) / \mathcal{H}_k \not\cong \F_q^* \times \PGL(\A). $$

\begin{example}
Let $ q=5 $, $ \A = \F_5^* $ and $ k = 2 $. In this case, the automorphism group of the Reed--Solomon code over $ \F_5^* $ of dimension $ 2 $ is isomorphic to
$$ \frac{\F_5^*\times \GL(\A)}
{\{(c^{-1},cI_2):c\in\F_5^*\}}
\cong
\GL(\A). $$
The center of this group is $ \{ cI_2 : c \in \F_5^* \} $, of size $ 4 $. On the other hand, $ \PGL(\A) $ is isomorphic to the dihedral group $ D_8 $. Thus
$$ \F_5^* \times \PGL(\A) \cong C_4 \times D_8, $$
whose center has size $ 4 \cdot 2 = 8 $. 
\end{example}

In some cases, however, we may obtain simpler groups. Next we give a sufficient condition to obtain the group $ \GL(\A) $.

\begin{proposition}
If $ q-1 $ and $ k-1 $ are coprime, then
$$ (\F_q^* \times \GL(\A)) / \mathcal{H}_k \cong \GL(\A). $$
\end{proposition}
\begin{proof}
We may choose the inverse $ d \in [q-2] $ of $ k-1 $ modulo $ q-1 $, and deduce that $ \mathcal{H}_k = \{ (c,c^{-d}I_2) : c \in \F_q^* \} $. Map $ (\lambda,M) $ to $ \lambda^d M $. This yields a surjective morphism from $ \F_q^* \times \GL(\A) $ to $ \GL(\A) $. Clearly its kernel is $ \mathcal{H}_k $, giving the desired isomorphism.
\end{proof}

When $ k-1 $ and $ q-1 $ share a common factor, the automorphism group may not be isomorphic to $ \GL(\A) $.

\begin{example}
Take $ q = 7 $, $ \A = \F_7^* $ and $ k = 3 $. In this case, $ \GL(\F_7^*) $ has an element of order $ 12 $:
$$ M = \left( \begin{array}{cc}
0&1\\3&0
\end{array} \right). $$
However, every element of the group $ (\F_7^* \times \GL(\F_7^*)) / \mathcal{H}_3 $ has order at most $ 6 $. This is because for every $ (\lambda, M) \in \F_7^* \times \GL(\F_7^*) $, $ M $ is diagonal or antidiagonal. In the first case, $ (\lambda^6,M^6) = (1,I_2) $, and in the second case, $ M^2 = d I_2 $, for some $ d \in \F_7^* $. Thus $ (\lambda^6,M^6) $ equals $ (1,d^3 I_2) $ modulo $ \mathcal{H}_3 $, which is the identity in the quotient since $ (1,d^3I_2) \in \mathcal{H}_3 $.
\end{example}

We conclude with the whole-field case ($ \A = \F_q $).

\begin{corollary} [\textbf{Whole field}]
Assume that $ q \geq 3 $, $ \A = \F_q $ and $ r \leq k \leq qr- r $. Assume also $ 2 \leq k \leq q-2 $ if $ r=1 $. We have that $ T \in \Aut(\M_{\F_q}(r,k)) $ if and only if there exist $ \lambda, \alpha \in \F_q^* $ and $ \beta \in \F_q $ such that
$$ T((z_a)_{a \in \F_q}) = \lambda (z_{\alpha a + \beta} \cdot \diag(1,\alpha,\ldots, \alpha^{r-1}))_{a \in \F_q}, $$
for $ (z_a)_{a \in \F_q} \in (\F_q^r)^q $. In particular, if $ {\rm AGL}(\F_q) $ denotes the group of affine maps of $ \F_q $, then 
$$ \Aut(\M_{\F_q}(r,k)) \cong \F_q^* \times {\rm AGL}(\F_q) . $$
\end{corollary}
\begin{proof}
The projectivity $ \psi_M $ in Theorem \ref{th NRT aut central range} associated to $ T $ satisfies $ \psi_M(\F_q) = \F_q $, thus it is affine (see Subsection \ref{subsec projective line}), i.e., we may choose $ \gamma = 0 $ and $ \delta = 1 $. As in Remark \ref{remark affine diagonal matrix}, in this case we have
$$ R^{(r)}_{a,M,k} = \diag(1, \alpha,\alpha^2, \ldots, \alpha^{r-1}), $$
for all $ a \in \F_q $, which gives the desired form of $ T $. Now in this case, for the affine map $ \psi_M $, there is only one matrix $ M \in \GL_2(\F_q) $ of the form
$$ M = \left( \begin{array}{cc}
\alpha & 0 \\
\beta & 1
\end{array} \right), $$
where $ \alpha \in \F_q^* $ and $ \beta \in \F_q $. Sending $ T $ to the pair $ (\lambda,\psi_{M^{-1}}) $ gives the desired group isomorphism ($ \psi_{M^{-1}} $ is chosen since $ \psi_N \circ \psi_M = \psi_{MN} $, thus $ \psi_{M^{-1}} \circ \psi_{N^{-1}} = \psi_{(MN)^{-1}} $, for $ M,N \in \GL_2(\F_q) $).
\end{proof}

\section{The permutation group} \label{sec permutation}

In this section, we compute the permutation group $ \PAut(\M_\A(r,k)) $ using the results from Section \ref{sec classification}. In contrast with the full NRT automorphism group (Theorems \ref{th NRT aut central range}, \ref{th NRT aut low range} and \ref{th NRT aut high range}), the group $ \PAut(\M_\A(r,k)) $ is quite different for the cases $ r = 1 $ and $ r \geq 2 $. In the first case (classical Reed--Solomon codes), it is formed by all affine maps that permute $ \A $ (see also \cite{CampsEtAl2025}), whereas in the pure multiplicity case ($ r \geq 2 $), it is formed by all affine translations that permute $ \A $.

For brevity, we will only treat the case of at least $ 3 $ points. We state the results for almost all dimensions ($ 2 \leq k \leq rn-2 $), but we will only prove the central case ($ r \leq k \leq rn-r $). The proof of the lower and higher ranges is analogous but using Theorems \ref{th NRT aut low range} and \ref{th NRT aut high range}.

\begin{theorem}  \label{th permutation group}
Assume that $ n = |\A| \geq 3 $ and $ 2 \leq k \leq nr- 2 $. A permutation $ \pi : \A \longrightarrow \A $ satisfies that $ P_\pi \in \PAut(\M_\A(r,k)) $ (Definition \ref{def Aut and PAut}) if and only if it is of the following form:
\begin{enumerate}
\item
If $ r=1 $: $ \pi(a) = \alpha a + \beta $, for all $ a \in \A $, for some $ \alpha \in \F_q^* $ and $ \beta \in \F_q $ (affinity).
\item
If $ r \geq 2 $: $ \pi(a) = a + \beta $, for all $ a \in \A $, for some $ \beta \in \F_q $ (translation).
\end{enumerate}
Notice that in all cases, $ \pi(\A) = \A $.
\end{theorem}
\begin{proof}
We first notice that $ P_\pi \in \PAut(\M_\A(r,k)) $ if and only if $ P_\pi^{-1} = P_{\pi^{-1}} \in \PAut(\M_\A(r,k)) $.

Given $ \beta \in \F_q $, it is easy to see that if $ f(x) \in \V_k $ and $ g(x) = f(x+\beta) \in \V_k $, then $ f^{(i)}(a + \beta) = g^{(i)}(a) $, for all $ i \in \mathbb{N} $ and $ a \in \A $ (Hasse derivatives commute with translations). Hence if $ \pi(a) = a + \beta $, for $ a \in \A $, then $ P_\pi^{-1} \in \PAut(\M_\A(r,k)) $, thus $ P_\pi \in \PAut(\M_\A(r,k)) $. For $ r=1 $, it is straightforward to verify that all affine permutations leave Reed--Solomon codes invariant.

Assume now $ r \leq k \leq rn-r $ (and $ 2 \leq k \leq n-2 $ if $ r=1 $). Let $ P_\pi \in \PAut(\M_\A(r,k)) $, hence $ P_{\pi^{-1}} = P_\pi^{-1} \in \PAut(\M_\A(r,k)) $. Since $ \PAut(\M_\A(r,k)) $ $ \subseteq \Aut(\M_\A(r,k)) $, by Theorem \ref{th NRT aut central range}, there exists $ \lambda \in \F_q^* $ and a projectivity $ \psi_M \in \PGL(\F_q) $ such that $ \psi_M(\A) = \A $ and
$$ P_\pi^{-1}((z_a)_{a \in \A}) = (z_{\pi(a)})_{a \in \A} = \lambda \left(z_{\psi_M(a)} R^{(r)}_{a,M,k} \right)_{a \in \A}, $$
for all $ (z_a)_{a \in \A} \in (\F_q^r)^n $. By uniqueness of $ \psi_M $ (see Theorem \ref{th NRT aut central range}), it must hold that $ \pi(a) = \psi_M(a) $, for $ a \in \A $. Using that the evaluation map $ {\rm ev} : \V_k \longrightarrow \F_q^r $ is surjective ($ k \geq r $), as in the proof of Theorem \ref{th NRT aut central range}, we deduce that
\begin{equation}
\lambda R^{(r)}_{a,M,k} = I_r ,
\label{eq lambda R = I_r for permutations}
\end{equation}
for all $ a \in \A $. 

Consider first the case $ r \geq 2 $. By Proposition \ref{prop projectivities are automorph}, looking at the first two diagonal elements, we have that
$$ \lambda (\gamma a + \delta)^{k-1} = 1 \quad \textrm{and} \quad \lambda \Delta(\gamma a + \delta)^{k-3} = 1. $$
Dividing, we obtain
$$ \frac{\Delta}{(\gamma a + \delta)^2} = 1, \quad \textrm{that is,} \quad (\gamma a + \delta)^2 - \Delta = 0. $$
This holds for all $ a \in \A $ and we are assuming $ n = |\A| \geq 3 $. Hence the polynomial $ (\gamma x + \delta)^2 - \Delta $ has degree at most $ 2 $ but at least $ 3 $ roots, thus it must be zero. Therefore $ \gamma = 0 $ and $ \Delta = \delta^2 $. Since $ \Delta = \alpha \delta - \beta \gamma \neq 0 $, we deduce that $ \delta \neq 0 $ and $ \delta^2 = \alpha \delta $. In other words, $ \delta = \alpha $. Since the projectivity does not change when multiplied by a nonzero scalar, we may choose $ \alpha=\delta=1 $, and the projectivity reads
$$ \psi_M(a) = \frac{\alpha a + \beta}{\gamma a + \delta} = a + \beta, $$
for all $ a \in \A $, which is of the desired form. Finally, for $ \gamma = 0 $ and $ \alpha=\delta=1 $, we have that $ R^{(r)}_{a,M,k} = I_r $ by Remark \ref{remark affine diagonal matrix}. Hence $ \lambda = 1 $ by Equation (\ref{eq lambda R = I_r for permutations}).

Finally, consider $ r= 1 $. In this case Equation~(\ref{eq lambda R = I_r for permutations}) becomes
$$ \lambda(\gamma a+\delta)^{k-1}=1, $$
for $ a \in \A $. Since $2\leq k\leq n-2$, the polynomial
\[
    (\gamma x+\delta)^{k-1}-\lambda^{-1}
\]
has degree at most $k-1<n$ and vanishes at every point of $\A$.
It is therefore the zero polynomial. Since $k-1>0$, its leading
coefficient gives $\gamma=0$. As before, $ \delta \neq 0 $, dividing by it and redefining $ \alpha $ and $ \beta $ as $ \alpha\delta^{-1} $ and $ \beta \delta^{-1} $, respectively, we have for $ a \in \A $,
$$ \pi(a) = \psi_M(a) = \alpha a + \beta. $$

For $ r \geq 3 $ and the ranges $ 2 \leq k \leq r-1 $ and $ nr-r+1 \leq k \leq nr-2 $, the proofs are similar, except we need to use Theorems \ref{th NRT aut low range} and \ref{th NRT aut high range}, respectively, from Appendix \ref{app low high range}. We leave the details to the reader. In the first case, use the first two entries of the diagonal and in the second case, the last two. In this way, one obtains 
$$ (\gamma a+\delta)^2=\Delta, $$
for all $ a \in \A $, and one proceeds as above.
\end{proof}

As groups, we have the following isomorphisms.

\begin{corollary}
Assume that $ n = |\A| \geq 3 $ and $ 2 \leq k \leq nr- 2 $. Define $ {\rm AGL}(\A) $ and $ {\rm T}(\A) $ as the groups of affine maps and translations of $ \F_q $, respectively, that leave $ \A $ invariant. Then we have natural group isomorphisms from Theorem \ref{th permutation group}:
\begin{enumerate}
\item
If $ r=1 $: $ \PAut(\M_\A(1,k)) \cong {\rm AGL}(\A) $.
\item
If $ r \geq 2 $: $ \PAut(\M_\A(r,k)) \cong {\rm T}(\A) $.
\end{enumerate}
\end{corollary}

\section{The folded Hamming metric} \label{sec folded Hamming}

Since multiplicity codes are sometimes considered with the folded Hamming metric \cite{KSY2014, martinez2025linear}, we conclude the manuscript with an algebraic description of the automorphisms of multiplicity codes in this metric. As we will see, all the NRT automorphisms found in Section \ref{sec classification} are also automorphisms for the folded Hamming metric. However, for the latter, there are many more.

\begin{definition}
We define the $ r $-folded Hamming weight in $ (\F_q^r)^n $ as
$$ \wt_{FH}((z_a)_{a \in \A}) = | \{ a \in \A : z_a \neq 0 \} |, $$
for all $ (z_a)_{a \in \A} \in (\F_q^r)^n $. For two $ c,c^\prime \in (\F_q^r)^n $, we define their $ r $-folded Hamming distance as $ \dd_{FH}(c,c^\prime) = \wt_{FH}(c-c^\prime) $. Notice that this is the same as the classical Hamming metric in $ (\F_q^r)^n $ of length $ n $ and alphabet $ \F_q^r $. 
\end{definition}

The linear isometries for the folded Hamming metric were characterized in \cite[Th. 2]{sr-hamming} (see also \cite[Prop. 5]{martinez2025linear}).

\begin{proposition} [\hspace{1pt}\cite{sr-hamming}] \label{prop isometries folded charact}
A linear map $ T : (\F_q^r)^n \longrightarrow (\F_q^r)^n $ is an isometry for the $ r $-folded Hamming metric if and only if there exists a permutation $ \pi : \A \longrightarrow \A $ and arbitrary invertible matrices $ A_a \in \GL_r(\F_q) $, for $ a \in \A $, such that 
$$ T((z_a)_{a \in \A}) = (z_{\pi(a)} A_a)_{a \in \A}, $$
for all $ (z_a)_{a \in \A} \in (\F_q^r)^n $. 
\end{proposition}

\begin{definition} [\textbf{Folded automorphism groups}] \label{def Aut and PAut folded}
We define the $ r $-folded automorphism group of $ \M_\A(r,k) $ as the group $ \FAut(\M_\A(r,k)) $ of linear isometries $ T : (\F_q^r)^n \longrightarrow (\F_q^r)^n $ for the $ r $-folded Hamming metric such that 
$$ T(\M_\A(r,k)) = \M_\A(r,k). $$
\end{definition}

Clearly, every linear isometry for the NRT metric is also an isometry for the folded Hamming metric (compare Propositions \ref{prop isometries NRT charact} and \ref{prop isometries folded charact}). In particular, we have the chain
$$ \PAut(\M_\A(r,k)) \subseteq \Aut(\M_\A(r,k)) \subseteq \FAut(\M_\A(r,k)), $$
and all the NRT automorphisms from Section \ref{sec classification} are also automorphisms of multiplicity codes for the folded Hamming metric. We now give an algebraic characterization of the latter. Again, for brevity, we will only treat the ``central dimension'' range for at least three evaluation points. Since clearly the NRT and $ r $-folded Hamming metrics both coincide with the classical Hamming metric for $ r=1 $, we only consider $ r \geq 2 $.

Consider as in Definition \ref{def K_a} the subspaces
$$ \K_a = \K_{a,r}=(x-a)^r\V_{k-r}, $$
for $a\in\A$. Denoting the tuple of such subspaces as $ \mathfrak K_{\A,r,k} = (\K_a)_{a\in\A} $, we define
$$ \Aut(\mathfrak K_{\A,r,k})
    =
    \left\{
        (L,\sigma)\in\GL(\V_k)\times\operatorname{Sym}(\A):
        L(\mathcal K_a)=\mathcal K_{\sigma(a)}
        \textrm{ for all }a\in\A
    \right\}, $$
where $ \GL(\V_k) $ is the group of vector space automorphisms of $ \V_k $ and $ \operatorname{Sym}(\A) $ is the symmetric group of $ \A $.

\begin{theorem} [\textbf{Folded Hamming automorphisms}] \label{th folded automorphisms}
Assume that $ n = |\A| \geq 3 $, $ r \geq 2 $ and $ r \leq k \leq nr- r $. Then there is a natural group isomorphism
\[
    \FAut(\M_\A(r,k))
    \cong
    \Aut(\mathfrak K_{\A,r,k}).
\]
More precisely, if $ T \in \FAut(\M_\A(r,k)) $ is given by
\[
    T((z_a)_{a\in\A})
    =
    (z_{\pi(a)}A_a)_{a\in\A},
\]
for all $ (z_a)_{a \in \A} \in (\F_q^r)^n $, for a permutation $ \pi $ of $ \A $ and $ A_a\in\GL_r(\F_q) $, for $ a \in \A $, and $L:\V_k\longrightarrow\V_k$ is the induced polynomial automorphism (adapt Proposition \ref{prop L linear op} to the folded Hamming metric), then the group isomorphism is given by
\[
    T\longmapsto(L,\pi^{-1}).
\]

If $k>r$, the subspaces $\mathcal K_a$, for $a\in\A$, are
distinct and therefore
\[
    \FAut(\M_\A(r,k))
    \cong \{ L \in \GL(\V_k) : L(\{ \K_a : a \in \A \}) = \{ \K_a : a \in \A \} \}.
\]
If $k=r$, then $\mathcal K_a = \{ 0 \} $, for all $ a \in \A $, and
\[
    \FAut(\M_\A(r,r))
    \cong
    \GL(\V_r)\times\operatorname{Sym}(\A).
\]
\end{theorem}
\begin{proof}
Let $ T \in\FAut(\M_\A(r,k)) $ be given by
\[
    T((z_a)_{a\in\A})
    =
    (z_{\pi(a)}A_a)_{a\in\A},
\]
for all $ (z_a)_{a \in \A} \in (\F_q^r)^n $. The vector space automorphism $ L : \V_k \longrightarrow \V_k $ acts as
\begin{equation}
J_a^{(r)}(Lf)
    =
    J_{\pi(a)}^{(r)}(f)A_a,
\label{eq L for folded}
\end{equation}
for all $ f \in \V_k $ and $ a \in \A $ (adapt Proposition \ref{prop L linear op} to the folded Hamming metric). Since $A_a$ is invertible, then we deduce from (\ref{eq L for folded}) that $ J_a^{(r)}(Lf) = 0 $ if and only if $ J_{\pi(a)}^{(r)}(f) = 0 $, which means that $ L(\K_{\pi(a)})=\K_a $. In other words, if $\sigma=\pi^{-1}$,
\[
    L(\K_b)=\K_{\sigma(b)},
\]
for all $ b \in \A $. Hence $(L,\sigma)\in\Aut(\mathfrak K_{\A,r,k})$.

Conversely, let $ (L,\sigma)\in\Aut(\mathfrak K_{\A,r,k}) $ and define $\pi=\sigma^{-1}$. Fix $ a \in \A $. Since $k\geq r$, the linear map
\[
    J_a^{(r)}:\V_k\longrightarrow\F_q^r
\]
is surjective and has kernel $\K_a = L(\K_{\pi(a)}) $. Hence there exists a unique matrix $ A_a \in \GL_r(\F_q) $ and induced vector space isomorphisms that make the following diagram commute:
\[
\begin{tikzcd}[column sep=large,row sep=large]
\V_k/\K_{\pi(a)}
    \arrow[r,"\overline{L}"]
    \arrow[d,"\overline{J}_{\pi(a)}^{(r)}"']
&
\V_k/\K_{a}
    \arrow[d,"\overline{J}_{a}^{(r)}"]
\\
\F_q^r
    \arrow[r,"\,A_a"']
&
\F_q^r ,
\end{tikzcd}
\]
where we identify the matrix $ A_a $ with the linear map that consists in multiplying by $ A_a $ (on the right, since we define $ \F_q^r = \F_q^{1 \times r} $). In other words, we have that
$$ J_a^{(r)}(Lf)
    =
    J_{\pi(a)}^{(r)}(f)A_a, $$
for every $f\in\V_k$. The resulting linear map $ T : (\F_q^r)^n \longrightarrow (\F_q^r)^n $, given by
\[
    T((z_a)_{a\in\A})
    =
    (z_{\pi(a)}A_a)_{a\in\A},
\]
for $ (z_a)_{a \in \A} \in (\F_q^r)^n $, is an isometry for the folded Hamming metric and, by construction,
$ T \in \FAut(\M_\A(r,k))$. This proves the first assertion.

The two maps above are group morphisms and inverses of each other. In the direction $ \FAut(\M_\A(r,k)) \longrightarrow \Aut(\mathfrak K_{\A,r,k}) $, it is because the permutation associated to $ T $ is unique and the induced map $ L $ is unique by injectivity of evaluation (Proposition \ref{prop L linear op}). In the opposite direction, it is because the matrices $ A_a $, for $ a \in \A $, that make the above diagrams commute are also unique.

Finally, if $k>r$, then $ \mathcal K_a=(x-a)^r\V_{k-r} $ is nonzero, and the spaces $\mathcal K_a$ are distinct, for $ a \in \A $.
Therefore the permutation of their labels is determined by $L$,
which gives the stated set-stabilizer description. If $k=r$, every
$\mathcal K_a$ is zero, and the last assertion follows immediately.
\end{proof}

We conclude by using the previous characterization to produce an example of an automorphism for the folded Hamming metric that is not an automorphism for the NRT metric.

\begin{example}
Let $\A=\F_3$ and $r=k=2$. Since
\[
    \K_a = \K_{a,2}=(x-a)^2\V_0=0,
\]
for every $a\in\F_3$, Theorem~\ref{th folded automorphisms} shows that
every $L\in\GL(\V_2)$, together with the identity permutation of
$\F_3$, induces a folded-Hamming automorphism of
$\M_{\F_3}(2,2)$.
Consider the vector space automorphism $ L : \V_2 \longrightarrow \V_2 $ given by
\[
    L(c_0+c_1x)
    =
    (c_0+c_1)+(c_0+2c_1)x,
\]
for $ c_0,c_1 \in \F_3 $. In the canonical ordered basis $ (1,x) $ of $ \V_2 $, the matrix of $L$ is 
\[
    H=
    \left( \begin{array}{cc}
     1&1\\
     1&2
    \end{array} \right) \in \GL_2(\F_3) .
\]
Fix $ a \in \F_3 $. The matrix of $ J_a^{(2)} : \V_2 \longrightarrow \F_3^2 $ in the ordered basis $ (1,x) $ and the canonical ordered basis of $ \F_3^2 $ is
\[
    E_a=
    \left( \begin{array}{cc}
    1&0\\
    a&1
    \end{array} \right).
\]
That is, $ J_a^{(2)}(f)=(c_0,c_1)E_a $, for $ (c_0,c_1) \in \F_3^2 $. Hence the corresponding matrix $ A_a \in \GL_2(\F_3) $ is
\[
    A_a=E_a^{-1}HE_a
    =
    \left( \begin{array}{cc}
        1+a&1\\
        1+a-a^2&2-a
    \end{array} \right).
\]
Explicitly,
\[
    A_0=
    \left( \begin{array}{cc}
        1&1\\
        1&2
    \end{array} \right),
    \quad
    A_1=
    \left( \begin{array}{cc}
        2&1\\
        1&1
    \end{array} \right) \quad \textrm{and} \quad
    A_2=
    \left( \begin{array}{cc}
        0&1\\
        2&0
    \end{array} \right).
\]
The corresponding automorphism $ T : (\F_3^2)^3 \longrightarrow (\F_3^2)^3 $, given by
\[
    T((z_a)_{a\in\F_3})=(z_aA_a)_{a\in\F_3},
\]
for $ (z_a)_{a \in \F_3} \in (\F_3^2)^3 $,
belongs to $\FAut(\M_{\F_3}(2,2))$. However, $A_0$ and $A_1$
have no zero entries, and in particular are neither upper nor lower
triangular. Thus $T$ is not an NRT automorphism, i.e., it does not belong to $\Aut(\M_{\F_3}(2,2))$.
\end{example}

\section*{Disclosure on AI usage}

The full permutation groups and partial results on the automorphism groups were first obtained by the authors. The idea of focusing on the NRT metric instead of the folded Hamming metric is due to the authors, after realizing the classical characterization does not hold for the folded Hamming metric. OpenAI's ChatGPT-5.6 Sol was used to fill the gaps in the proof to obtain the NRT automorphism group in the general case. AI was not used for the writing of the article. The authors take responsibility for the correctness of the results.

\section*{Acknowledgements}

This work was started while H. H. L{\'o}pez, R. San Jos{\'e} and E. Camps Moreno were visiting the IMUVa-Mathematics Research Institute, University of Valladolid, Spain.

This work has been partially supported by Grant PID2022-138906NB-C21 funded by MICIU/AEI/ 10.13039/501100011033 and by ERDF/EU and Grant CLU-2025-1-02-IMUVA by the Department of Education of the Junta de Castilla y León and FEDER Funds. Hiram H. L\'opez was partially supported by the NSF grants DMS-2401558 and DMS-2502705 and the Commonwealth Cyber Initiative. Rodrigo San-Jos\'e was also partially supported by the NSF grant DMS-2401558, the Commonwealth Cyber Initiative, and an AMS-Simons Travel Grant.

\bibliographystyle{plain}

\begin{appendix}

\section*{Appendices}

\section{The low and high dimensional cases} \label{app low high range}

In this appendix, we treat the ``low and high'' dimensional ranges. In these cases, the automorphisms have more degrees of freedom. They cannot be characterized simply by a projectivity up to a nonzero scalar.

\begin{theorem} [\textbf{Low range}] \label{th NRT aut low range}
Assume that $ n = |\A| \geq 3 $, $ r \geq 3 $ and $ 2 \leq k \leq r-1 $. A linear isometry $ T : (\F_q^r)^n \longrightarrow (\F_q^r)^n $ for the NRT metric lies in $ \Aut(\M_\A(r,k)) $ if and only if there exists $ \lambda \in \F_q^* $, a projectivity $ \psi_M \in \PGL(\F_q) $ such that $ \psi_M(\A) = \A $, and arbitrary upper triangular matrices $ B_a \in \GL_{r-k}(\F_q) $, for $ a \in \A $, such that 
$$ T((z_a)_{a \in \A}) = (z_{\psi_M(a)} A_a)_{a \in A}, $$
for all $ (z_a)_{a \in \A} \in (\F_q^r)^n $, and for the upper triangular matrices
$$ A_a = \left( \begin{array}{cc}
\lambda R^{(k)}_{a,M,k} & 0 \\
0 & B_a 
\end{array} \right) \in \GL_r(\F_q), $$
for all $ a \in \A $, where the invertible upper triangular matrices $ R^{(k)}_{a,M,k} \in \F_q^{k \times k} $ are given in Proposition \ref{prop projectivities are automorph}. Finally, the projectivity $ \psi_M $ is unique.
\end{theorem} 
\begin{proof}
The proof is essentially the same as in Theorem \ref{th NRT aut central range}, since we may still use Proposition \ref{prop pi is a projectivity} and Theorem \ref{th automorph with pi = id is lambda}, because we are assuming that $ n \geq 3 $ and $ 2 \leq k \leq r-1 $.

Let $ T \in \Aut(\M_\A(r,k)) $. Hence there exist $ \lambda \in \F_q^* $, a projectivity $ \psi_M \in \PGL(\F_q) $ such that $ \psi_M(\A) = \A $, and with associated map $ L : \V_k \longrightarrow \V_k $, as in Proposition \ref{prop L linear op}, which acts as
$$ L(f)=\lambda\Psi_{M,k}(f), $$
for $ f \in \V_k $. However, notice that now if $ f \in \V_k $, since $ k \leq r-1 $, then $ f^{(i)} = 0 $, for all $ i \in [k,r] $ (which is a nonempty interval). Thus for any $ (z_a)_{a \in \A} \in \M_\A(r,k) $, the $ a $th block is
$$ z_a = (z_{a,0}, z_{a,1}, \ldots, z_{a,k-1}, 0,0, \ldots, 0) \in \F_q^r. $$
For $ a \in \A $, we may subdivide the upper triangular matrix $ A_a \in \GL_r(\F_q) $ associated to $ T $ as
$$ A_a = \left( \begin{array}{cc}
M_a & N_a \\
0 & B_a
\end{array} \right), $$
for some upper triangular $ M_a \in \GL_k(\F_q) $ and $ B_a \in \GL_{r-k}(\F_q) $, and some arbitrary matrix $ N_a \in \F_q^{k \times (r-k)} $.  By Proposition~\ref{prop L linear op},
$$ J_a^{(k)}(Lf) = J_{\psi_M(a)}^{(k)}(f)M_a. $$
On the other hand, Proposition~\ref{prop projectivities are automorph}
gives
$$ J_a^{(k)}(Lf) = \lambda J_{\psi_M(a)}^{(k)}(f) R_{a,M,k}^{(k)}. $$
Now in contrast with Theorem \ref{th NRT aut central range}, the map $ J_{\psi_M(a)}^{(r)}:\V_k \longrightarrow \F_q^k \times \{0 \}^{r-k} $ is a vector space isomorphism. Hence we conclude that 
$$ M_a=\lambda R_{a,M,k}^{(k)} \quad \textrm{and} \quad N_a = 0, $$
but $ B_a $ may be arbitrary (since it only interacts with the last $ r-k $ zeros of every block), for all $ a \in \A $. Conversely, by the same arguments, any map $ T $ of this form belongs to $ \Aut(\M_\A(r,k)) $.
\end{proof} 

The ``high-dimensional'' range equally has more degrees of freedom than the central range. Since we were assuming $ 2 \leq k \leq rn - r $ in most of Sections \ref{sec permutation is projectivity} and \ref{sec d degree graph}, we will use a duality argument to go from Theorem \ref{th NRT aut low range} (the low range) to the high range.

Duals of univariate multiplicity codes were first given in \cite{RVW2024}. See also \cite[Prop. B.3]{DualsMultiplicity2026}. They have also been independently obtained later in \cite{LiEtAl2026}. We consider duals with respect to the usual $ \F_q $-linear inner product of $ (\F_q^r)^n = \F_q^{rn} $.

\begin{lemma}[\hspace{1pt}\cite{RVW2024,DualsMultiplicity2026}] \label{lemma duals multiplicity}
Let $ h_{a,i} \in \V_{rn} $ be the unique polynomial in $ \V_{rn} $ such that
$$ h_{a,i}^{(j)}(b)=\delta_{a,b}\delta_{i,j}, $$
where $ \delta_{a,b} $ and $ \delta_{i,j} $ denote the Kronecker delta, for all $ a,b \in \A $ and $ i,j \in [0,r-1] $. Define $ \lambda_{a,i} \in \F_q $ be the coefficient of $ x^{rn-1} $ in $ h_{a,i} $, for all $ a \in \A $ and $ i \in [0,r-1] $. Consider the upper anti-triangular matrix $ M_a \in \GL_r(\F_q) $ whose $ (i,j) $th entry is $ \lambda_{a,i+j} $ if $ i+j \leq r-1 $ and zero otherwise, for $ i,j \in [0,r-1] $. Then
$$ \M_\A(r,k)^\perp = \varphi_\A ( \M_\A(r,rn-k)), $$
where $ \varphi_\A( (z_a)_{a \in \A} ) = (z_a M_a)_{a \in \A} $, for all $ (z_a)_{a \in \A} \in (\F_q^r)^n $.
\end{lemma}

Define now the $ r \times r $ invertible matrices over $ \F_q $,
$$ R_r = \left( \begin{array}{ccccc}
0 & 0 & \ldots & 0 & 1 \\
0 & 0 & \ldots & 1 & 0 \\
\vdots & \vdots & & \vdots & \vdots \\
0 & 1 & \ldots & 0 & 0 \\
1 & 0 & \ldots & 0 & 0
\end{array} \right) \quad \textrm{and} \quad U_a = M_a R_r = \left( \begin{array}{ccccc}
\lambda_{a,r-1} & \lambda_{a,r-2} & \ldots & \lambda_{a,1} & \lambda_{a,0} \\
0 & \lambda_{a,r-1} & \ldots & \lambda_{a,2} & \lambda_{a,1} \\
\vdots & \vdots & \ddots & \vdots & \vdots \\
0 & 0 & \ldots & \lambda_{a,r-1} & \lambda_{a,r-2} \\
0 & 0 & \ldots & 0 & \lambda_{a,r-1}
\end{array} \right), $$
for $ a \in \A $. Notice that $U_a$ is an upper triangular Toeplitz matrix. In particular, for all $ a \in \A $,
$$ R_rU_a^\intercal R_r=U_a \quad \textrm{and} \quad R_rU_a^{-\intercal}R_r=U_a^{-1}. $$
Define $ s = nr-k \in [2,r-1] $ and $ t = r-s \in [1, r-2] $. For all $ a \in\A $, write
$$ U_a = \left( \begin{array}{cc}
P_a & Q_a\\
0 & S_a
\end{array} \right), $$
where $ P_a \in \GL_t(\F_q) $ and $ S_a \in \GL_s(\F_q) $ are upper triangular, and define
$$ W_{a,s} = Q_aS_a^{-1} \in \F_q^{t \times s}. $$

\begin{theorem} [\textbf{High range}] \label{th NRT aut high range}
Assume that $ n = |\A| \geq 3 $, $ r \geq 3 $ and $ nr-r +1 \leq k \leq nr-2 $. A linear isometry $ T : (\F_q^r)^n \longrightarrow (\F_q^r)^n $ for the NRT metric lies in $ \Aut(\M_\A(r,k)) $ if and only if there exists $ \lambda \in \F_q^* $, a projectivity $ \psi_M \in \PGL(\F_q) $ such that $ \psi_M(\A) = \A $, and arbitrary upper triangular matrices $ B_a \in \GL_t(\F_q) $, for $ a \in \A $, such that 
$$ T((z_a)_{a \in \A}) = (z_{\psi_M(a)} A_a)_{a \in A}, $$
for all $ (z_a)_{a \in \A} \in (\F_q^r)^n $, and for the upper triangular matrices
\begin{equation}
A_a = \lambda R^{(r)}_{a,M,k} \left( \begin{array}{cc}
B_a & (I_t - B_a)W_{a,s} \\
0 & I_s 
\end{array} \right) \in \GL_r(\F_q),
\label{eq A_a for high range}
\end{equation}
for all $ a \in \A $, where the invertible upper triangular matrices $ R^{(r)}_{a,M,k} \in \F_q^{r \times r} $ are given in Proposition \ref{prop projectivities are automorph}. Finally, the projectivity $ \psi_M $ is unique.
\end{theorem}
\begin{proof}
Let $ T : (\F_q^r)^n \longrightarrow (\F_q^r)^n $ be a linear isometry for the NRT metric with permutation $ \pi : \A \longrightarrow \A $ and upper triangular matrices $ A_a \in \GL_r(\F_q) $, for $ a \in \A $, as in Proposition \ref{prop isometries NRT charact}. Let $ \theta(T) : (\F_q^r)^n \longrightarrow (\F_q^r)^n $ be the linear isometry with the same permutation $ \pi $ but upper triangular matrices
\begin{equation}
C_a = U_{\pi(a)}R_r(A_a^{-1})^\intercal R_rU_a^{-1}.
\label{eq C_a}
\end{equation}
We have that $ T(\C)=\C $ if and only if $ (T^{-1})^\intercal(\C^\perp)=\C^\perp $, where $ (T^{-1})^\intercal $ also has permutation $ \pi $ but has upper triangular matrices $ (A_a^{-1})^\intercal $, for $ a \in \A $. Together with Lemma \ref{lemma duals multiplicity}, we have that
\begin{equation}
\theta : \Aut(\M_\A(r,k)) \longrightarrow \Aut(\M_\A(r,rn-k))
\label{eq theta}
\end{equation}
is a group isomorphism. In particular, $ \theta $ preserves the permutation $ \pi $.

Let now $T\in\Aut(\M_\A(r,k))$. Since $s=rn-k$ satisfies 
$2\leq s<r$, Theorem~\ref{th NRT aut low range}, applied to
$\theta(T)$, shows that its permutation $ \pi $ is given by a unique
projectivity $\psi_M$ such that $ \psi_M(\A) = \A $.

Let $T_M \in \Aut(\M_\A(r,k))$ be as in Corollary~\ref{cor T_M given by projectivity}.
Since $\theta$ preserves permutations,
$\theta(T_M)$ also has permutation $\pi = \psi_M$. By Theorem \ref{th NRT aut low range} (the low range case), its action on the first $s$ coordinates in each block in $ \F_q^r $ is a global
nonzero scalar multiple of $R_{a,M,s}^{(s)}$. Hence, after multiplying
$T_M$ by a suitable global scalar $\lambda\in\F_q^*$, the automorphism
$$ H=T\circ(\lambda T_M)^{-1} \in \Aut(\M_\A(r,k)) $$
has identity block permutation and $\theta(H)$ has upper triangular matrices
of the form
\[
    \left( \begin{array}{cc}
        I_s&0\\
        0&C_a
        \end{array} \right),
\]
where $C_a \in \GL_t(\F_q)$ is an arbitrary invertible upper triangular matrix.

For a linear isometry with identity permutation, Equation (\ref{eq C_a}) gives
\[
\left( \begin{array}{cc}
        I_s&0\\
        0&C_a
        \end{array} \right)
    =
    U_aR_r(H_a^{-1})^\intercal R_rU_a^{-1},
\]
where $H_a \in \GL_r(\F_q) $ is the upper triangular matrix of $H$ at $ a \in \A $. Since
$R_rU_a^\intercal R_r=U_a$, we deduce that
\[
    H_a
    =
    U_aR_r
    \left( \begin{array}{cc}
        I_s&0\\
        0&(C_a^{-1})^\intercal
        \end{array} \right)
    R_rU_a^{-1}.
\]
After reversing the two diagonal blocks,
\[
    H_a
    =
    U_a
    \left( \begin{array}{cc}
        D_a&0\\
        0&I_s
        \end{array} \right)
    U_a^{-1},
\]
where $ D_a=R_t(C_a^{-1})^\intercal R_t \in \GL_t(\F_q) $ is equally an arbitrary invertible upper triangular matrix. Using
\[
    U_a=
    \begin{pmatrix}
        P_a&Q_a\\
        0&S_a
    \end{pmatrix},
\]
a direct multiplication yields
\[
    U_a
    \begin{pmatrix}
        D_a&0\\
        0&I_s
    \end{pmatrix}
    U_a^{-1}
    =
    \begin{pmatrix}
        P_aD_aP_a^{-1}
        &
        (I_t-P_aD_aP_a^{-1})Q_aS_a^{-1}\\
        0&I_s
    \end{pmatrix}.
\]
Conjugation by $P_a$ is a bijection of the group of invertible upper
triangular $t\times t$ matrices. Renaming
\[
    B_a=P_aD_aP_a^{-1},
    \qquad
    W_{a,s}=Q_aS_a^{-1},
\]
we obtain
\[
    H_a=
    \begin{pmatrix}
        B_a&(I_t-B_a)W_{a,s}\\
        0&I_s
    \end{pmatrix}.
\]
Finally,
\[
    T=H\circ(\lambda T_M),
\]
and therefore its local matrix is exactly (\ref{eq A_a for high range}).

Every step above is reversible, which proves the converse as well.
Uniqueness of the projectivity follows from the corresponding
uniqueness in the low-dimensional case (Theorem \ref{th NRT aut low range}).
\end{proof}

\section{The cases $ k=1 $, $ k=rn-1 $, $ n=1 $ and $ n= 2 $} \label{app exceptional cases}

In this appendix, we treat the cases $ k=1 $, $ k=rn-1 $, $ n=1 $ and $ n= 2 $, which are not treated by Theorems \ref{th NRT aut central range}, \ref{th NRT aut low range} and \ref{th NRT aut high range}.

\begin{proposition} [\textbf{Dimension one}] \label{prop aut dimension one}
Assume $ k = 1 $. Then $ \Aut(\M_\A(r,1)) $ consists of the maps $ T : (\F_q^r)^n \longrightarrow (\F_q^r)^n $ given by
$$ T((z_a)_{a \in \A}) = \lambda \left( z_{\pi (a)} \left( \begin{array}{cc}
1 & 0 \\
0 & B_a
\end{array} \right) \right)_{a \in \A}, $$
for all $ (z_a)_{a \in \A} \in (\F_q^r)^n $, where $ \lambda \in \F_q^* $, the matrices $ B_a \in \GL_{r-1}(\F_q) $ are arbitrary upper triangular, and the permutation $ \pi : \A \longrightarrow \A $ is arbitrary.
\end{proposition}
\begin{proof}
Notice in this case 
$$ \M_\A(r,1) = \{ ((c,0,\ldots,0))_{a \in \A} : c \in \F_q \}. $$
Clearly all of the maps $ T $ as in the proposition lie in $ \Aut(\M_\A(r,1)) $. Now if $ T \in \Aut(\M_\A(r,1)) $, then the first rows of the upper triangular matrices $ A_a \in \GL_r(\F_q) $ from Proposition \ref{prop isometries NRT charact} must be of the form $ (\lambda, 0,\ldots,0) \in \F_q^r $ (independent of $ a $), for $ a \in \A $. Factoring out the scalar $ \lambda \in \F_q^* $, we obtain the desired expression for $ T $.  
\end{proof}

\begin{proposition}[\textbf{Codimension one}] \label{prop aut codimension one}
Assume that $n\geq2$ and $k=rn-1$.
Then $\Aut(\M_\A(r,rn-1))$ consists exactly of the maps $ T : (\F_q^r)^n \longrightarrow (\F_q^r)^n $ given by $ T((z_a)_{a\in\A})
    =
    (z_{\pi(a)}A_a)_{a\in\A} $, for $ (z_a)_{a \in \A} \in (\F_q^r)^n $, where the permutation $\pi:\A\longrightarrow\A$ is arbitrary and, for $ a \in \A $,
\[
    A_a
    =
    \lambda U_{\pi(a)}
    \left( \begin{array}{cc}
    B_a&0\\
        0&1
    \end{array} \right)
    U_a^{-1},
    \qquad
    \lambda\in\F_q^*,
\]
for some $ \lambda \in \F_q^* $ and some arbitrary upper triangular matrix
$B_a\in\GL_{r-1}(\F_q)$.

For $r=1$, this reduces to the classical Reed--Solomon formula
\[
    A_a
    =
    \lambda\frac{G_\A'(a)}{G_\A'(\pi(a))},
\]
where $ G_\A = \prod_{a \in \A} (x-a) \in \F_q[x] $ is the vanishing polynomial of $ \A $.
\end{proposition}
\begin{proof}
Apply the complementary-dimension isomorphism from the proof of
Theorem~\ref{th NRT aut high range} to the code $ \M_\A(r,1) $.
By Proposition~\ref{prop aut dimension one}, an automorphism of
$\M_\A(r,1)$ has arbitrary block permutation $\pi$ and upper triangular matrices
\[
\left( \begin{array}{cc}
\mu&0\\
        0&C_a
\end{array} \right) \in \GL_r(\F_q),
\]
where $\mu\in\F_q^*$ is independent of $a$ and
$C_a\in\GL_{r-1}(\F_q)$ is arbitrary upper triangular.

Applying the inverse of the group isomorphism $ \theta : \Aut(\M_\A(r,rn-1)) \longrightarrow \Aut(\M_\A(r,1)) $ from (\ref{eq theta}) and reversing the two
diagonal blocks gives
\[
    U_{\pi(a)}
    \begin{pmatrix}
        R_{r-1}(C_a^{-1})^\intercal R_{r-1}&0\\
        0&\mu^{-1}
    \end{pmatrix}
    U_a^{-1}.
\]
Put $\lambda=\mu^{-1}$. Notice that $ R_{r-1}(C^{-1})^\intercal R_{r-1} $ may be written as $\lambda B_a$, where $ B_a \in \GL_{r-1}(\F_q) $ is an arbitrary invertible upper triangular matrix.
This proves the displayed formula. 

For the case $ r=1 $, simply notice that then, for all $ a \in \A $,
\[
    U_a=M_a=(G_\A'(a)^{-1}).
\]
\end{proof}

\begin{proposition} [\textbf{One point}]
Let $ \A = \{ a \} $ ($ n=1 $) and $ 1 \leq k \leq r $. Then $ \Aut(\M_\A(r,k)) $ consists of the maps $ T : (\F_q^r)^n \longrightarrow (\F_q^r)^n $ given by
$$ T(z_a) = z_a \left( \begin{array}{cc}
B & 0 \\
0 & C
\end{array} \right), $$
for all $ z_a \in \F_q^r $, where the matrices $ B \in \GL_k(\F_q) $ and $ C \in \GL_{r-k}(\F_q) $ are arbitrary upper triangular (and the permutation $ \pi $ is clearly the identity).
\end{proposition}
\begin{proof}
Simply notice that in this case $ \M_\A(r,k) = \F_q^k \times \{ 0 \}^{r-k} $.
\end{proof}

For brevity, we do not make the two-point case fully explicit. The reader may obtain explicit formulas for the corresponding upper triangular matrices. We also treat only the case $ 1 \leq k \leq r $ and leave the case $ r < k \leq 2r $ to the reader.

\begin{proposition}[\textbf{Two points}]
\label{prop two points}
Let $\A=\{a,b\} \subseteq \F_q$, with $a\neq b$, and let $ k \in [r] $. Define
$$ p_j(x)=(x-a)^j(x-b)^{k-1-j} \in \V_k, $$
for $ j\in[0,k-1] $. Then a linear NRT isometry $T : (\F_q^r)^2 \longrightarrow (\F_q^r)^2 $ belongs to $\Aut(\M_\A(r,k))$ if and only if
its permutation $\pi$ is arbitrary and
the induced polynomial operator $L:\V_k\longrightarrow\V_k$ (Proposition \ref{prop L linear op}) satisfies
\[
    L(p_j)=
    \begin{cases}
        \lambda_jp_j& \textrm{ if } \pi={\rm Id}_\A,\\
        \lambda_jp_{k-1-j}& \textrm{ if } \pi(a)=b,\ \pi(b)=a,
    \end{cases}
\]
for $ j\in[0,k-1] $, for arbitrary scalars $\lambda_0,\lambda_1,\ldots,\lambda_{k-1}\in\F_q^*$.
On the first $k$ coordinates of each block, the corresponding
upper triangular matrices are uniquely determined by $L$, whereas,
if $k<r$, arbitrary invertible upper triangular matrices
may act on the last $r-k$ coordinates.
\end{proposition}
\begin{proof}
The polynomials
$p_0,p_1,\ldots,p_{k-1}$ form a basis of $\V_k$. If the block permutation
is the identity, Lemma~\ref{lemma L(K) is K pi -1} shows that $L$ preserves the
two complete flags
\[
    \K_{a,0}\supsetneq\K_{a,1}\supsetneq\cdots\supsetneq\K_{a,k}
    \quad\textrm{and}\quad
    \K_{b,0}\supsetneq\K_{b,1}\supsetneq\cdots\supsetneq\K_{b,k}.
\]
For every $j\in[0,k-1]$, we have that
\[
    \K_{a,j}\cap\K_{b,k-1-j}
    =
    \langle p_j\rangle_{\F_q}.
\]
Hence each $p_j$ is an eigenvector of $L$, giving $ L(p_j)=\lambda_jp_j $, for arbitrary $\lambda_j\in\F_q^*$.

If $\pi$ exchanges $a$ and $b$, Lemma~\ref{lemma L(K) is K pi -1} instead gives
\[
    L\bigl(
    \K_{a,j}\cap\K_{b,k-1-j}
    \bigr)
    =
    \K_{a,k-1-j}\cap\K_{b,j},
\]
and therefore $ L(p_j)=\lambda_jp_{k-1-j} $, for arbitrary $\lambda_j\in\F_q^*$.

Conversely, every operator of either displayed form sends the two
multiplicity flags according to the prescribed permutation.
Since
\[
    J_a^{(k)},J_b^{(k)}:\V_k\longrightarrow\F_q^k
\]
are vector space isomorphisms, it induces unique invertible upper
triangular matrices on the first $k$ coordinates at both blocks.
If $k<r$, the remaining coordinates vanish on every codeword, so
arbitrary invertible upper triangular matrices may act independently
on those coordinates.
\end{proof}

\end{appendix}

\end{document}